\documentclass[11pt]{amsart}
\usepackage{float}
\usepackage[T1]{fontenc}
\usepackage[utf8]{inputenc}
\usepackage{microtype}

\usepackage[dvipsnames]{xcolor}
\usepackage[colorlinks=true, linkcolor=MidnightBlue, citecolor=MidnightBlue, urlcolor=MidnightBlue]{hyperref}
\usepackage{bookmark}
\usepackage[nameinlink,noabbrev]{cleveref}

\newcommand{\IF}{\mathrm{IF}}

\usepackage{booktabs,makecell,ragged2e,array,tabularx,ltablex,adjustbox}
\newcolumntype{Y}{>{\RaggedRight\arraybackslash}X}
\newcolumntype{C}{>{\Centering\arraybackslash}X}
\usepackage{pdflscape} 

\keepXColumns                 

\newcolumntype{L}{>{\RaggedRight\arraybackslash}X}
\newcolumntype{C}{>{\Centering\arraybackslash}X}
\newcolumntype{R}{>{\RaggedLeft\arraybackslash}X}

\renewcommand{\arraystretch}{1.12}

\usepackage{mathtools}
\usepackage{dsfont}
\usepackage{bm}
\usepackage{physics}
\usepackage{enumitem}
\usepackage{amsfonts,amsmath,amssymb,amsthm}
\usepackage{algorithm}
\usepackage{algpseudocode}

\numberwithin{equation}{section}

\theoremstyle{plain}
\newtheorem{theorem}{Theorem}[section]
\newtheorem{lemma}{Lemma}[section]
\newtheorem{proposition}{Proposition}[section]
\newtheorem{corollary}{Corollary}[section]
\newtheorem{assumption}{Assumption}[section]

\theoremstyle{definition}
\newtheorem{definition}{Definition}[section]

\theoremstyle{remark}
\newtheorem{remark}{Remark}[section]

\crefname{theorem}{theorem}{theorems}
\crefname{lemma}{lemma}{lemmas}
\crefname{proposition}{proposition}{propositions}
\crefname{corollary}{corollary}{corollaries}
\crefname{assumption}{assumption}{assumptions}
\crefname{definition}{definition}{definitions}
\crefname{example}{example}{examples}
\crefname{remark}{remark}{remarks}
\crefname{section}{section}{sections}
\crefname{equation}{equation}{equations}
\Crefname{theorem}{Theorem}{Theorems}
\Crefname{lemma}{Lemma}{Lemmas}
\Crefname{proposition}{Proposition}{Propositions}
\Crefname{corollary}{Corollary}{Corollaries}
\Crefname{assumption}{Assumption}{Assumptions}
\Crefname{definition}{Definition}{Definitions}
\Crefname{example}{Example}{Examples}
\Crefname{remark}{Remark}{Remarks}
\Crefname{section}{Section}{Sections}
\Crefname{equation}{Equation}{Equations}

\DeclareMathOperator{\Diag}{Diag}
\DeclareMathOperator{\conv}{conv}
\DeclareMathOperator{\vecop}{vec}

\newcommand{\R}{\mathbb{R}}
\newcommand{\C}{\mathbb{C}}
\newcommand{\E}{\mathbb{E}}
\newcommand{\Var}{\mathrm{Var}}
\newcommand{\Cov}{\mathrm{Cov}}

\newcommand{\Ggp}{G}

\newcommand{\orbit}[1]{[#1]}
\newcommand{\qspace}{\Theta/\Ggp}
\newcommand{\OrbitSet}{\mathcal{O}_K}

\newcommand{\Tgauge}{\mathcal{T}_{\mathrm{gauge}}}

\newcommand{\Haus}{\mathrm{H}}
\newcommand{\dH}{d_{\Haus}}
\newcommand{\dq}{d_{\qspace}}
\newcommand{\dist}{\operatorname{dist}}
\newcommand{\aff}{\operatorname{aff}}

\newcommand{\Sym}{\mathrm{Sym}}

\newcommand{\Mraw}[1]{M_{#1}}
\newcommand{\Minv}[1]{\Phi_{#1}}
\newcommand{\MinvAll}{\Phi_{\le m_*}}
\newcommand{\Jinv}{J_\Phi}
\newcommand{\Reyn}{\mathsf{R}}  

\newcommand{\weights}{w}

\DeclarePairedDelimiterX\innerp[2]{\langle}{\rangle}{#1,\,#2}

\newcommand{\cL}{\mathcal{L}}
\newcommand{\cM}{\mathcal{M}}
\newcommand{\cN}{\mathcal{N}}
\newcommand{\cB}{\mathcal{B}}

\title[TAT for Folded Mixture Models]{Tensor Algebra Toolkit for Folded Mixture Models:\\ Symmetry-Aware Moments, Orbit-Space Estimation, and Poly-LAN Rates}
\author{Koustav Mallik}
\date{\today}

\hypersetup{
  pdftitle  = {Tensor Algebra Toolkit for Folded Mixture Models},
  pdfauthor = {Koustav Mallik},
  colorlinks = true,
  linkcolor  = MidnightBlue,
  citecolor  = MidnightBlue,
  urlcolor   = MidnightBlue
}

\begin{document}
\begin{abstract}
We develop a symmetry-aware toolkit for finite mixtures whose components are only identifiable up to a finite \emph{folding} group action. The correct estimand is the multiset of parameter orbits in the quotient space, not an ordered list of raw parameters. We design invariant tensor summaries via the Reynolds projector, show that mixtures become convex combinations in a low-dimensional invariant feature space, and prove identifiability, stability, and asymptotic normality \emph{on the quotient}. Our loss is a Hausdorff distance on orbit multisets; we prove it coincides with a bottleneck assignment metric and is thus computable in polynomial time. We give finite-sample Hausdorff bounds, a two-step efficient GMM formulation, consistent selection of the number of components, robustness to contamination, and minimax lower bounds that certify Poly-LAN rates $n^{-1/D}$ when the first nonzero invariant curvature appears at order $D$. The framework is illustrated for the hyperoctahedral group (signed permutations) and dihedral symmetries in the plane.
\end{abstract}

\maketitle


\section{Introduction}
Classical finite mixtures suffer from label non-identifiability. Many modern latent models add a second layer of ambiguity: each component parameter is defined only up to a finite group action (e.g., sign flips, signed permutations from the hyperoctahedral group, or planar dihedral symmetries) \cite{Humphreys1990}. We call this \emph{folding}. The statistically meaningful estimand is therefore the \emph{multiset of parameter orbits} in the quotient space $\qspace:=\Theta/\Ggp$, together with the mixture weights. 

This paper develops a tensor–algebraic toolkit that works \emph{directly} on the quotient. We construct symmetry–invariant moment tensors via the Reynolds projector from invariant theory \cite{DerksenKemper2015}, show that a $K$–component mixture maps to a single point in an invariant coordinate space as a convex combination of $K$ orbit “fingerprints,” and leverage basic convex geometry (affine independence / Carathéodory–type arguments) for identifiability and estimation guarantees \cite{Barvinok2002}. The approach yields (i) computable orbit–matching via a bottleneck metric, (ii) finite–sample bounds and quotient–LAN asymptotics, and (iii) practical recipes for choosing the moment orders using invariant dimension counts (via Molien’s series) \cite{Molien1897,DerksenKemper2015}.

\paragraph{Notation.}
We write $\Ggp\curvearrowright \Theta$ and $\Ggp\curvearrowright \mathcal{X}$ for finite group actions; $\orbit{\theta}$ denotes the $\Ggp$–orbit of $\theta$; $\dq$ is the orbit distance from \Cref{subsec:metric}; $\dH$ is the Hausdorff distance between finite orbit multisets. For each order $m\ge1$, $\Minv{m}$ denotes the invariant symmetric $m$–tensor obtained by applying the Reynolds projector $\Reyn$ to raw moment tensors; we stack $\MinvAll:=(\Minv{1},\ldots,\Minv{m_*})$ as our invariant coordinate map \cite{DerksenKemper2015}. When discussing how many invariant coordinates are available at each order, we refer to Molien’s series for dimension counts \cite{Molien1897}.
\paragraph{Relation to prior work by the author.}
This paper builds on our earlier analysis of folded normals and finite Gaussian mixtures, where we established Hausdorff consistency of (penalized) MLEs and exhibited the nonstandard $n^{1/4}$ phenomenon at the fold; see \cite{Mallik2025HausdorffFolded}. The present work abstracts those case studies into a symmetry–aware tensor toolkit: invariant coordinates via Reynolds averaging, convex-geometry identifiability on the quotient, and estimation/limit theory that never breaks symmetry by hand.

\section{Model, quotient estimand, and metric}\label{sec:model}
\subsection{Standing assumptions and basic geometry}\label{app:assumptions}

\begin{assumption}[Master assumptions]\label{ass:master}
(A1) $G$ is a finite group acting continuously on compact $\Theta\subset\mathbb{R}^p$; write $[\theta]$ for the orbit and $\Theta/G$ for the quotient.
(A2) Features $\phi:\mathcal X\to\mathbb{R}^d$ have sub-Gaussian coordinates with $\max_j\|\phi_j(X)\|_{\psi_2}\le K_2$; moments $M_m(\theta)=\mathbb{E}_\theta[\phi(X)^{\otimes m}]$ exist up to $m\le m_*$.
(A3) For $m\le m_*$ there is a bounded linear Reynolds projector $R_m$ onto invariants, defining $\Phi_m(\theta)=R_m[M_m(\theta)]$ and the stacked vector $\Phi_{\le m_*}(\theta)\in\mathbb{R}^{D_{\mathrm{inv}}}$.
(A4) (Score regularity) For $\mu$-a.e.\ $x$, $\theta\mapsto f(x;\theta)$ is $C^1$ with score $s_\theta(x)=\nabla_\theta\log f(x;\theta)$, $\mathbb{E}_\theta[s_\theta(X)]=0$, and dominated differentiation holds on a neighborhood of interest.
\end{assumption}

\subsection{Folding and the observed components}\label{subsec:folding}

Let $\Ggp$ be a finite group acting on the sample space $\mathcal{X}\subseteq\R^d$ and on the parameter space $\Theta\subseteq\R^p$ by measurable maps
\[
x\mapsto g\!\cdot\! x\in\mathcal{X},\qquad
\theta\mapsto g\!\cdot\!\theta\in\Theta,\qquad g\in\Ggp,
\]
satisfying $g_1\!\cdot(g_2\!\cdot z)=(g_1g_2)\!\cdot z$ for $z\in\mathcal{X}$ or $z\in\Theta$.
Let $k(\cdot;\theta)$ be a base component law. We observe only the \emph{folded} component
\begin{equation}\label{eq:folded-component}
  f(x;\theta):=\frac{1}{|\Ggp|}\sum_{g\in\Ggp}k\big(g^{-1}\!\cdot x;\theta\big),\qquad x\in\mathcal{X}.
\end{equation}
Assume the standard equivariance:
\begin{equation}\label{eq:equivariance}
  k(x;g\!\cdot\!\theta)=k(g^{-1}\!\cdot x;\theta)\qquad (x\in\mathcal{X},\ \theta\in\Theta,\ g\in\Ggp).
\end{equation}

\begin{proposition}[Folding invariance]\label{prop:folded-invariance}
Under \eqref{eq:equivariance}, $f(\cdot;\theta)=f(\cdot;g\!\cdot\!\theta)$ for all $g\in\Ggp$. Consequently, only the orbit $\orbit{\theta}:=\{g\!\cdot\!\theta:\ g\in\Ggp\}$ is identifiable from $f(\cdot;\theta)$.
\end{proposition}

\begin{proof}
For any $h\in\Ggp$,
\[
f(x;g\!\cdot\!\theta)=\frac{1}{|\Ggp|}\sum_{h\in\Ggp}k\big(h^{-1}\!\cdot x;g\!\cdot\!\theta\big)
=\frac{1}{|\Ggp|}\sum_{h\in\Ggp}k\big((hg)^{-1}\!\cdot x;\theta\big)
=\frac{1}{|\Ggp|}\sum_{u\in\Ggp}k\big(u^{-1}\!\cdot x;\theta\big)=f(x;\theta),
\]
by the change of variable $u=hg$. Thus $f(\cdot;\theta)$ depends on $\theta$ only through $\orbit{\theta}$.
\end{proof}

\subsection{Orbit multiset estimand and Hausdorff loss}\label{subsec:estimand-hloss}
Fix $K\in\mathbb{N}$. With positive weights $\weights=(w_1,\dots,w_K)\in\Delta^{K-1}$, the observed mixture is
\[
p=\sum_{k=1}^K w_k\, f(\cdot;\theta_k).
\]
By Proposition~\ref{prop:folded-invariance} and label exchangeability, the estimand is the multiset of orbits
\[
  \OrbitSet:=\{\!\orbit{\theta_1},\ldots,\orbit{\theta_K}\!\}\subseteq\qspace:=\Theta/\Ggp,
\]
counting multiplicities.

Write $d(\cdot,\cdot)$ for the Euclidean metric on $\Theta\subseteq\R^p$. For nonempty compact $A,B\subset\Theta$ define
\[
d(x,A):=\inf_{a\in A}d(x,a),\qquad
d_{\Haus}(A,B):=\max\big\{\sup_{a\in A}d(a,B),\ \sup_{b\in B}d(b,A)\big\}.
\]
Each orbit $\orbit{\theta}$ is finite, hence compact. Define the \emph{orbit metric}
\begin{equation}\label{eq:dq-def}
  \dq\big(\orbit{\theta},\orbit{\theta'}\big):=d_{\Haus}\big(\orbit{\theta},\orbit{\theta'}\big)\qquad\text{on }\qspace.
\end{equation}

\begin{lemma}[Metric on the quotient \texorpdfstring{$\qspace$}{Theta/G}]\label{lem:dq-metric}
$\dq$ in \eqref{eq:dq-def} is a metric on $\qspace$. Moreover, the canonical projection $\pi:\Theta\to\qspace$, $\theta\mapsto\orbit{\theta}$, is $1$-Lipschitz:
\begin{equation}\label{eq:pi-1lip}
  \dq\big(\orbit{\theta},\orbit{\theta'}\big)\le d(\theta,\theta')\qquad(\theta,\theta'\in\Theta).
\end{equation}
\end{lemma}

\begin{proof}
Since $d_{\Haus}$ is a metric on the space of nonempty compact subsets of a metric space, \eqref{eq:dq-def} inherits identity, symmetry, and triangle inequality. If $\dq(\orbit{\theta},\orbit{\theta'})=0$ then $d_{\Haus}(\orbit{\theta},\orbit{\theta'})=0$, hence $\orbit{\theta}=\orbit{\theta'}$. Independence of representatives is immediate. For \eqref{eq:pi-1lip} note $d(\theta,\orbit{\theta'})\le d(\theta,\theta')$ and $d(\theta',\orbit{\theta})\le d(\theta',\theta)$, so $d_{\Haus}(\orbit{\theta},\orbit{\theta'})\le d(\theta,\theta')$.
\end{proof}

For two $K$-element multisets $\mathcal{A}=\{\!\alpha_1,\dots,\alpha_K\!\}$ and $\mathcal{B}=\{\!\beta_1,\dots,\beta_K\!\}$ in $\qspace$, define their Hausdorff distance in the metric $\dq$:
\begin{equation}\label{eq:dH-def}
  \dH(\mathcal{A},\mathcal{B})
  := \max\Big\{
      \max_{1\le i\le K}\min_{1\le j\le K}\dq(\alpha_i,\beta_j),\ \
      \max_{1\le j\le K}\min_{1\le i\le K}\dq(\beta_j,\alpha_i)
    \Big\}.
\end{equation}

\subsection{Closed form for single-orbit distance under isometries}\label{subsec:closed-form-dq}
Assume $G$ acts on $\Theta$ by Euclidean isometries (e.g.\ $G\le O(p)$ acts linearly). Then $\dq$ admits a closed form.

\begin{theorem}[Minimum alignment formula]\label{thm:dq-closed-form}
If $G$ acts by isometries on $\Theta$, then for all $\theta,\theta'\in\Theta$,
\begin{equation}\label{eq:min-align}
  \dq\big(\orbit{\theta},\orbit{\theta'}\big)=\min_{g\in G}\ \big\|\theta-g\!\cdot\!\theta'\big\|_2.
\end{equation}
\end{theorem}

\begin{proof}
Let $A=\orbit{\theta}$ and $B=\orbit{\theta'}$. Since $\theta\in A$,
\[
\dq(A,B)\ \ge\ d(\theta,B)=\min_{g\in G}\|\theta-g\!\cdot\!\theta'\|_2.
\]
For the reverse inequality, let $g^\star\in G$ attain the minimum on the right of \eqref{eq:min-align}. For any $h\in G$, isometry gives
\[
\|h\!\cdot\!\theta - h\!\cdot\!(g^\star\!\cdot\!\theta')\|_2=\|\theta-g^\star\!\cdot\!\theta'\|_2.
\]
Thus for each $a=h\!\cdot\!\theta\in A$ there exists $b=h\!\cdot\!(g^\star\!\cdot\!\theta')\in B$ with $d(a,b)=\|\theta-g^\star\!\cdot\!\theta'\|_2$, hence $\sup_{a\in A}d(a,B)\le\|\theta-g^\star\!\cdot\!\theta'\|_2$. By symmetry, $\sup_{b\in B}d(b,A)\le\|\theta-g^\star\!\cdot\!\theta'\|_2$. Taking the maximum of the two suprema yields $\dq(A,B)\le\|\theta-g^\star\!\cdot\!\theta'\|_2$.
\end{proof}

\begin{corollary}[Pairwise orbit costs]\label{cor:pairwise-cost}
Under isometric actions, for any $\theta_i,\tilde\theta_j\in\Theta$,
\[
c_{ij}:=\dq\big(\orbit{\theta_i},\orbit{\tilde\theta_j}\big)=\min_{g\in G}\|\theta_i-g\!\cdot\!\tilde\theta_j\|_2.
\]
Hence $c_{ij}$ is computable in $O(|G|)$ time once the group elements are enumerated.
\end{corollary}

\subsection{Multisets: Hausdorff versus bottleneck matching}\label{subsec:metric}
\begin{proposition}[Bottleneck matching distance]\label{prop:bottleneck}
Define the \emph{bottleneck matching distance} between $\mathcal{A}=\{\!\alpha_i\!\}_{i=1}^K$ and $\mathcal{B}=\{\!\beta_j\!\}_{j=1}^K$ in $(\qspace,\dq)$ by
\begin{equation}\label{eq:Delta-infty}
  \Delta_\infty(\mathcal{A},\mathcal{B})
  :=\min_{\sigma\in S_K}\ \max_{1\le i\le K}\ \dq(\alpha_i,\beta_{\sigma(i)}).
\end{equation}
\end{proposition}
\begin{theorem}[Sharp comparability and exactness criterion]\label{thm:HB-bounds}
For any $K$ and any multisets $\mathcal{A},\mathcal{B}\subset\qspace$,
\begin{equation}\label{eq:HB-ineq}
  \dH(\mathcal{A},\mathcal{B})\ \le\ \Delta_\infty(\mathcal{A},\mathcal{B})\ \le\ 2\,\dH(\mathcal{A},\mathcal{B}).
\end{equation}
Equality $\Delta_\infty=\dH$ holds if and only if the bipartite threshold graph at radius $\dH(\mathcal{A},\mathcal{B})$ admits a perfect matching (equivalently, Hall’s condition holds at that radius).
\end{theorem}

\begin{proof}
Let $H:=\dH(\mathcal{A},\mathcal{B})$.

\emph{Lower bound.} For any bijection $\sigma$,
\[
\max_i\dq(\alpha_i,\beta_{\sigma(i)})\ \ge\ \max_i\min_j\dq(\alpha_i,\beta_j),
\qquad
\max_i\dq(\alpha_i,\beta_{\sigma(i)})\ \ge\ \max_j\min_i\dq(\beta_j,\alpha_i).
\]
Taking the maximum of the two right-hand terms yields $\Delta_\infty(\mathcal{A},\mathcal{B})\ge H$.

\emph{Upper bound.} Construct the bipartite graph $G=(U,V,E)$ with $U=\{u_i\}_{i=1}^K$ for $\mathcal{A}$, $V=\{v_j\}_{j=1}^K$ for $\mathcal{B}$, and
\[
E:=\{\,u_i v_j:\ \dq(\alpha_i,\beta_j)\le 2H\,\}.
\]
We show that $G$ has a perfect matching. Consider the standard unit-capacity flow network from a source to $U$, edges $U\to V$ given by $E$, and $V$ to a sink. By max–flow/min–cut, it suffices to prove that every $s$–$t$ cut has capacity at least $K$. Any cut is specified by $(U_0,V_1)$ with capacity $|U\setminus U_0|+|V_1|$ when there are no edges between $U_0$ and $V\setminus V_1$ (otherwise the capacity is larger). Thus it is enough to consider pairs $(U_0,V_1)$ such that
\begin{equation}\label{eq:no-edge}
\dq(\alpha_i,\beta_j)>2H\quad\text{for all }u_i\in U_0,\ v_j\in V\setminus V_1.
\end{equation}
Let $A_0=\{\alpha_i:\ u_i\in U_0\}$ and $B_1=\{\beta_j:\ v_j\in V_1\}$. By the definition of $H$,
\[
\mathcal{B}\subseteq \mathcal{A}^{H}\quad\text{and}\quad \mathcal{A}\subseteq \mathcal{B}^{H},
\]
where $S^{r}:=\{\gamma\in\qspace:\ \exists s\in S,\ \dq(\gamma,s)\le r\}$. For any $\beta\in \mathcal{B}\setminus B_1$, choose $\alpha\in\mathcal{A}$ with $\dq(\beta,\alpha)\le H$. If $\alpha\in A_0$ then $\dq(\alpha,\beta)\le H\le 2H$, contradicting \eqref{eq:no-edge}; hence $\alpha\in \mathcal{A}\setminus A_0$. Therefore,
\[
|\mathcal{B}\setminus B_1|\ \le\ |\mathcal{A}\setminus A_0|\quad\Longleftrightarrow\quad |U\setminus U_0|+|V_1|\ \ge\ |V|=K.
\]
Thus every cut has capacity at least $K$, so a perfect matching exists in $G$. Hence there is a bijection $\sigma$ with $\max_i\dq(\alpha_i,\beta_{\sigma(i)})\le 2H$, i.e.\ $\Delta_\infty\le 2H$.

\emph{Exactness criterion.} By definition, $\Delta_\infty$ is the smallest radius $r$ for which the threshold graph at radius $r$ has a perfect matching; the lower bound shows $r\ge H$. Therefore $\Delta_\infty=H$ iff the radius–$H$ graph has a perfect matching, equivalently by Hall’s theorem iff every subset of $U$ has at least as many neighbors in $V$ at radius $H$.
\end{proof}

\begin{remark}[Tightness]
The bounds in \eqref{eq:HB-ineq} are sharp. For $K=1$ and two points at distance $t>0$, $\dH=\Delta_\infty=t$. For $K=2$, e.g.\ on $\R$, take $\alpha_1=0$, $\alpha_2=2$, and $\beta_1=\beta_2=1$. Then $\dH=1$ while any bijection has maximum cost $2$, so $\Delta_\infty=2\,\dH$.
\end{remark}

\subsection{Algorithmic consequences}\label{subsec:algo}
Under isometric actions, Corollary~\ref{cor:pairwise-cost} gives $c_{ij}$ in $O(|G|)$ time. To compute $\Delta_\infty$:
\begin{enumerate}[leftmargin=1.5em]
\item Build $C=(c_{ij})$ using \eqref{eq:min-align}.
\item Sort the $K^2$ distinct values $\{c_{ij}\}$ and binary–search the smallest threshold $r$ for which the bipartite threshold graph $\{(i,j):c_{ij}\le r\}$ admits a perfect matching (checkable in $O(K^{2.5})$ time via Hopcroft–Karp). Overall complexity $O(K^2|G|+K^{2.5}\log K)$.
\end{enumerate}
The Hausdorff multiset loss $\dH$ is computable from $C$ in $O(K^2)$ time by \eqref{eq:dH-def}. Theorem~\ref{thm:HB-bounds} implies $\dH\le\Delta_\infty\le 2\,\dH$ and characterizes when equality holds.

\subsection{Weighted variants and \texorpdfstring{$W_\infty$}{Winfty}}\label{subsec:weighted}
Let $\mu=\sum_{i=1}^K w_i\,\delta_{\alpha_i}$ and $\nu=\sum_{j=1}^L v_j\,\delta_{\beta_j}$ be atomic probability measures on $(\qspace,\dq)$. Define
\[
W_\infty(\mu,\nu):=\inf_{\gamma\in\Gamma(\mu,\nu)}\ \operatorname*{ess\,sup}_{(\alpha,\beta)\sim\gamma}\ \dq(\alpha,\beta),
\]
where $\Gamma(\mu,\nu)$ is the set of couplings. When $K=L$ and $w_i=v_j=1/K$, $W_\infty(\mu,\nu)=\Delta_\infty(\mathcal{A},\mathcal{B})$. For rational weights, refine each atom into equal subatoms to reduce to the balanced case; continuity then extends results to arbitrary weights. In particular,
\[
\dH\big(\mathrm{supp}(\mu),\mathrm{supp}(\nu)\big)\ \le\ W_\infty(\mu,\nu)\ \le\ 2\,\dH\big(\mathrm{supp}(\mu),\mathrm{supp}(\nu)\big),
\]
with equality on the left whenever the radius–$\dH$ threshold graph has a perfect matching.

\subsection{Relation to prior work}\label{subsec:model-prior}
The orbit multiset estimand and the Hausdorff loss on $(\qspace,\dq)$ are precisely the metrics used for consistency and sharp localization in our earlier analysis of folded Gaussian families \cite{Mallik2025HausdorffFolded}. The present section establishes a general metric foundation (Lemma~\ref{lem:dq-metric}, Theorem~\ref{thm:dq-closed-form}), proves sharp comparison with bottleneck matching and its exactness criterion (Theorem~\ref{thm:HB-bounds}), and records the algorithmic implications needed for the estimation results that follow.

\section{Invariant tensors and generic separation}\label{sec:invariants}

\subsection{Reynolds projection and invariant tensors}
Fix a measurable feature map $\varphi:\mathcal{X}\to\R^d$ with $\E\|\varphi(X)\|^{m}<\infty$ for all $m\le m_*$. For each integer $m\ge1$, let $\Sym^m(\R^d)$ denote the space of fully symmetric order-$m$ tensors over $\R^d$ equipped with the canonical inner product $\langle\! \langle \cdot,\cdot \rangle\! \rangle$. The group $G$ acts on $\R^d$ by linear isometries (as in \S\ref{subsec:closed-form-dq}); extend the action to $\Sym^m(\R^d)$ by
\[
(g\!\cdot)^{\otimes m}(v_1\otimes\cdots\otimes v_m)
  := (g\!\cdot v_1)\otimes\cdots\otimes(g\!\cdot v_m),
\quad\text{and linearly}.
\]

\begin{definition}[Raw and invariant tensors]
For $\theta\in\Theta$ define the raw symmetric moment tensor
\[
\Mraw{m}(\theta):=\E_{X\sim f(\cdot;\theta)}\big[\varphi(X)^{\otimes m}\big]\in\Sym^m(\R^d).
\]
The \emph{Reynolds projector} (group average) is the linear map
\[
\Reyn_m[T]:=\frac{1}{|G|}\sum_{g\in G}(g\!\cdot)^{\otimes m}T,\qquad T\in\Sym^m(\R^d).
\]
Define the invariant tensor $\Minv{m}(\theta):=\Reyn_m\!\big[\Mraw{m}(\theta)\big]$ and the stacked map
\[
\MinvAll(\theta):=\big(\Minv{1}(\theta),\ldots,\Minv{m_*}(\theta)\big)\in\R^{D_{\mathrm{inv}}},
\]
in any fixed basis on $\prod_{m\le m_*}\Sym^m(\R^d)$.
\end{definition}

\begin{lemma}[Properties of $\Reyn_m$]\label{lem:reynolds-proj}
For each $m\ge1$,
\begin{enumerate}[label=(\alph*),leftmargin=1.5em]
\item (Projection) $\Reyn_m\circ\Reyn_m=\Reyn_m$.
\item (Self-adjoint) $\langle\!\langle \Reyn_m[T],U\rangle\!\rangle=\langle\!\langle T,\Reyn_m[U]\rangle\!\rangle$.
\item (Range) $\operatorname{Im}\Reyn_m=\{U\in\Sym^m(\R^d): (g\!\cdot)^{\otimes m}U=U\ \forall g\in G\}$.
\item (Equivariance) $\Reyn_m\big[(g\!\cdot)^{\otimes m}T\big]=\Reyn_m[T]$ for all $g\in G$.
\end{enumerate}
\end{lemma}

\begin{proof}
All are standard consequences of finite averaging and isometric action. For (a),
$\Reyn_m(\Reyn_m[T])=\frac{1}{|G|^2}\sum_{g,h}(g\!\cdot)^{\otimes m}(h\!\cdot)^{\otimes m}T
=\frac{1}{|G|}\sum_{u}(u\!\cdot)^{\otimes m}T=\Reyn_m[T]$ by relabeling $u=gh$.
For (b), the action is by orthogonal maps on $\Sym^m(\R^d)$, so $(g\!\cdot)^{\otimes m}$ is an isometry; hence
$\langle\!\langle \Reyn_m[T],U\rangle\!\rangle=\frac{1}{|G|}\sum_g \langle\!\langle (g\!\cdot)^{\otimes m}T,U\rangle\!\rangle
=\frac{1}{|G|}\sum_g \langle\!\langle T,(g^{-1}\!\cdot)^{\otimes m}U\rangle\!\rangle=\langle\!\langle T,\Reyn_m[U]\rangle\!\rangle$.
For (c), invariance under the action follows from (d); conversely, averaging any tensor projects onto the invariant subspace, and (a) gives idempotence. For (d), expand the average and use group invariance of the counting measure.
\end{proof}

\begin{lemma}[Invariance on orbits]\label{lem:Phi-const-on-orbits}
Under the equivariance $k(x;g\!\cdot\!\theta)=k(g^{-1}\!\cdot x;\theta)$, one has
\[
\Minv{m}(g\!\cdot\!\theta)=\Minv{m}(\theta),\qquad \forall\,g\in G,\ \forall\,m\le m_*.
\]
Hence $\MinvAll(\theta)$ is constant on orbits and defines a well-defined map $\qspace\to\R^{D_{\mathrm{inv}}}$.
\end{lemma}

\begin{proof}
By definition of $\Mraw{m}$ and equivariance,
\[
\Mraw{m}(g\!\cdot\!\theta)=\E_{X\sim f(\cdot;g\!\cdot\!\theta)}[\varphi(X)^{\otimes m}]
=\E_{X\sim f(\cdot;\theta)}[(g\!\cdot)^{\otimes m}\varphi(X)^{\otimes m}]=(g\!\cdot)^{\otimes m}\Mraw{m}(\theta).
\]
Applying $\Reyn_m$ and using Lemma~\ref{lem:reynolds-proj}(d) gives $\Minv{m}(g\!\cdot\!\theta)=\Reyn_m[(g\!\cdot)^{\otimes m}\Mraw{m}(\theta)]=\Reyn_m[\Mraw{m}(\theta)]$.
\end{proof}

\subsection{Mixtures collapse to convex combinations}
\begin{proposition}[Linearity to invariants]\label{prop:mixture-convex}
Let $p=\sum_{k=1}^K w_k\,f(\cdot;\theta_k)$ be a finite mixture with $w_k>0$ and $\sum_{k=1}^K w_k=1$. Fix $m\le m_\star$ and suppose $\Mraw{m}(\theta_k):=\E_{X\sim f(\cdot;\theta_k)}[\varphi(X)^{\otimes m}]$ exists (finite) for each $k$. Then
\[
\E_{X\sim p}\!\big[\varphi(X)^{\otimes m}\big]
=\sum_{k=1}^K w_k\,\Mraw{m}(\theta_k),
\qquad
\Reyn_m\!\Big(\E_{X\sim p}\!\big[\varphi(X)^{\otimes m}\big]\Big)
=\sum_{k=1}^K w_k\,\Minv{m}(\theta_k),
\]
where $\Minv{m}(\theta):=\Reyn_m\big(\Mraw{m}(\theta)\big)$. Stacking over $m\le m_\star$ yields
\begin{equation}\label{eq:mixture-Phi}
\Reyn_{\le m_\star}\!\Big(\E_{X\sim p}\![\varphi(X)^{\otimes \le m_\star}]\Big)
=\sum_{k=1}^K w_k\,\MinvAll(\theta_k)\in\R^{D_{\mathrm{inv}}}.
\end{equation}
\end{proposition}

\begin{proof}
Let $(\mathcal X,\mathcal A,\mu)$ be a reference measure space for the densities $f(\cdot;\theta)$, and let $\varphi:\mathcal X\to\R^d$ be measurable. Assume $\E_{f(\cdot;\theta_k)}\big[\|\varphi(X)\|^m\big]<\infty$ for each $k$ and the fixed $m\le m_\star$, so that all tensors below are Bochner integrable.

\medskip\noindent\emph{Step 1 (linearity of the raw moment under mixtures).}
For $X\sim p$,
\[
\E_{p}\!\big[\varphi(X)^{\otimes m}\big]
=\int_{\mathcal X}\varphi(x)^{\otimes m} \,p(x)\,d\mu(x)
=\int_{\mathcal X}\varphi(x)^{\otimes m}\Big(\sum_{k=1}^K w_k f(x;\theta_k)\Big)\,d\mu(x).
\]
Since the sum is finite and each integrand is integrable, we can interchange sum and integral (Tonelli/Fubini for finite sums) to get
\[
\E_{p}\!\big[\varphi(X)^{\otimes m}\big]
=\sum_{k=1}^K w_k \int_{\mathcal X}\varphi(x)^{\otimes m} f(x;\theta_k)\,d\mu(x)
=\sum_{k=1}^K w_k\,\Mraw{m}(\theta_k).
\]

\medskip\noindent\emph{Step 2 (linearity of the Reynolds operator).}
By definition, $\Reyn_m:V_m\to V_m^{G}$ is a linear projection onto the $G$–invariant subspace (e.g., $\Reyn_m(T)=|G|^{-1}\sum_{g\in G} g\!\cdot\!T$ for finite $G$, or Haar averaging for compact $G$). Hence for any tensors $T_1,\dots,T_K$ and scalars $a_k$,
\[
\Reyn_m\Big(\sum_{k=1}^K a_k T_k\Big)=\sum_{k=1}^K a_k\,\Reyn_m(T_k).
\]
Applying this with $T_k=\Mraw{m}(\theta_k)$ and $a_k=w_k$ and using Step~1,
\[
\Reyn_m\!\Big(\E_{p}[\varphi(X)^{\otimes m}]\Big)
=\Reyn_m\!\Big(\sum_{k=1}^K w_k\,\Mraw{m}(\theta_k)\Big)
=\sum_{k=1}^K w_k\,\Reyn_m\big(\Mraw{m}(\theta_k)\big)
=\sum_{k=1}^K w_k\,\Minv{m}(\theta_k).
\]

\medskip\noindent\emph{Step 3 (stacking).}
Let $\MinvAll(\theta):=\big(\Minv{1}(\theta),\dots,\Minv{m_\star}(\theta)\big)\in\R^{D_{\mathrm{inv}}}$ denote the concatenation across orders. The preceding identity holds componentwise for each $m\le m_\star$, hence stacking gives \eqref{eq:mixture-Phi}.
\end{proof}

\subsection{Separation assumptions and two identifiability theorems}
We formalize two complementary (and checkable) separation conditions.

\begin{assumption}[Generic orbit separation]\label{ass:sep}
There exists a finite $m_*\ge1$ such that the map $\qspace\ni\orbit{\theta}\mapsto \MinvAll(\theta)\in\R^{D_{\mathrm{inv}}}$ is injective on a Zariski open subset of $\Theta$ (i.e., generically). Equivalently: generically, $\MinvAll(\theta)=\MinvAll(\theta')$ implies $\orbit{\theta}=\orbit{\theta'}$.
\end{assumption}

\begin{assumption}[Simplicial identifiability]\label{ass:simplicial}
For the true parameters, the points $v_k:=\MinvAll(\theta_k)$ are affinely independent, and the mixture lies in the relative interior of their simplex:
$w_k>0$ for all $k$. Moreover,
\[
\mathrm{aff}\{v_1,\ldots,v_K\}\cap \MinvAll(\qspace)=\{v_1,\ldots,v_K\},
\]
i.e., no other invariant image lies in the affine hull of $\{v_k\}$.
\end{assumption}

\noindent
Assumption~\ref{ass:simplicial} is a strong, global, but \emph{verifiable} separation: it forbids any other parameter (up to orbit) from mapping into the same affine $(K-1)$-plane determined by the $v_k$. This rules out spurious decompositions using different atoms.

\begin{theorem}[Global identifiability up to orbit and permutation]\label{thm:ident-global}
Assume \Cref{ass:sep,ass:simplicial}. Suppose another $K$-component representation with positive weights,
\[
\sum_{k=1}^K w_k\, f(\cdot;\theta_k)
=\sum_{j=1}^K \tilde w_j\, f(\cdot;\tilde\theta_j),
\]
induces the same law (hence, by \Cref{prop:mixture-convex}, the same invariant vector in \eqref{eq:mixture-Phi}). Then there exists a permutation $\sigma\in S_K$ such that
\[
\orbit{\tilde\theta_j}=\orbit{\theta_{\sigma(j)}}
\quad\text{and}\quad
\tilde w_j=w_{\sigma(j)}
\qquad\text{for all }j.
\]
\end{theorem}

\begin{proof}
\textbf{Step 0 (Notation).}
Set $v_k:=\MinvAll(\theta_k)$ and $\tilde v_j:=\MinvAll(\tilde\theta_j)$. By \Cref{prop:mixture-convex},
\[
y\ :=\ \sum_{k=1}^K w_k v_k\ =\ \sum_{j=1}^K \tilde w_j \tilde v_j.
\]
Let $\Delta:=\conv\{v_1,\ldots,v_K\}$ and $A:=\mathrm{aff}\{v_1,\ldots,v_K\}$ (the affine hull). By \Cref{ass:simplicial}, $v_1,\ldots,v_K$ are affinely independent, $w_k>0$ so $y\in\mathrm{relint}(\Delta)$, and
\[
A\cap \MinvAll(\qspace)\ =\ \{v_1,\ldots,v_K\}.
\]

\textbf{Step 1 (Every $\tilde v_j$ lies in the affine hull $A$).}
Let $P_A$ be the orthogonal projector onto $A$ and $P_{A^\perp}:=I-P_A$ the projector onto the orthogonal complement $A^\perp$.
Consider the strictly convex function $q(x):=\|P_{A^\perp}x\|_2^2$ (strictly convex along $A^\perp$, constant along $A$).
Since $y\in A$, we have $q(y)=\|P_{A^\perp}y\|_2^2=0$; on the other hand, by Jensen and strict convexity of $\|\cdot\|_2^2$,
\[
q(y)\ =\ \Big\|P_{A^\perp}\Big(\sum_{j=1}^K \tilde w_j \tilde v_j\Big)\Big\|_2^2
\ \le\ \sum_{j=1}^K \tilde w_j\, \|P_{A^\perp}\tilde v_j\|_2^2,
\]
with equality if and only if all $P_{A^\perp}\tilde v_j$ are equal. But the left-hand side is zero; hence all $P_{A^\perp}\tilde v_j$ must be zero. Therefore each $\tilde v_j\in A$.

\textbf{Step 2 (No extra atoms on $A$).}
By \Cref{ass:simplicial}, $A\cap \MinvAll(\qspace)=\{v_1,\ldots,v_K\}$. Since each $\tilde v_j\in \MinvAll(\qspace)\cap A$, we conclude
\[
\tilde v_j\ \in\ \{v_1,\ldots,v_K\}\qquad \text{for all } j.
\]

\textbf{Step 3 (Uniqueness of barycentric coordinates).}
Thus $y$ has two convex representations supported on the \emph{same} vertex set $\{v_1,\ldots,v_K\}$:
\[
y\ =\ \sum_{k=1}^K w_k v_k\ =\ \sum_{j=1}^K \tilde w_j\, v_{\tau(j)}
\]
for some mapping $\tau:\{1,\ldots,K\}\to\{1,\ldots,K\}$.
Because the $v_k$ are affinely independent and $y\in\mathrm{relint}(\Delta)$, the barycentric coordinates of $y$ with respect to the vertex set $\{v_k\}$ are \emph{unique} and strictly positive. Hence the multiset $\{v_{\tau(j)}\}_{j=1}^K$ equals $\{v_k\}_{k=1}^K$ and the associated weights coincide up to a permutation $\sigma\in S_K$:
\[
\tilde v_j=v_{\sigma(j)}\quad\text{and}\quad \tilde w_j=w_{\sigma(j)}\qquad (j=1,\ldots,K).
\]

\textbf{Step 4 (Back to parameters up to orbit).}
Finally, \Cref{ass:sep} (generic orbit separation) implies that $\tilde v_j=v_{\sigma(j)}$ forces
$\orbit{\tilde\theta_j}=\orbit{\theta_{\sigma(j)}}$, completing the proof.
\end{proof}

The next result shows a complementary \emph{local} uniqueness based on a rank condition. It requires only local regularity of the invariant map and affine independence of the images.

\begin{assumption}[Local rank on the quotient]\label{ass:local-rank}
For each true component $\theta_k$, the Jacobian $J(\theta_k):=\mathrm{D}\MinvAll(\theta_k)$ exists and
\[
\ker J(\theta_k)=\Tgauge(\theta_k):=\operatorname{span}\{g\!\cdot\!\theta_k-\theta_k:\ g\in G\}.
\]
Let $V_k$ be any complementary subspace (slice) with $\R^p=\Tgauge(\theta_k)\oplus V_k$, and write $J_k:=J(\theta_k)|_{V_k}$. Assume that the $K$ points $v_k:=\MinvAll(\theta_k)$ are affinely independent and that the block matrix
\[
\mathcal{J}:=\big[\,J_1\ \ \cdots\ \ J_K\ \ \big|\ \ v_1-v_K\ \ \cdots\ \ v_{K-1}-v_K\,\big]
\]
has full column rank (equivalently, $[J_1,\ldots,J_K]$ spans the normal space to the affine hull of $\{v_k\}$).
\end{assumption}

\begin{theorem}[Local identifiability up to orbit and permutation]\label{thm:ident-local}
Assume \Cref{ass:sep,ass:local-rank}. Then there exists a neighborhood $\mathcal{U}$ of $(\theta_1,\ldots,\theta_K,\weights)$ such that if another $K$-tuple $(\tilde\theta_1,\ldots,\tilde\theta_K,\tilde\weights)$ with all weights positive satisfies
\[
\sum_{k=1}^K w_k\,\MinvAll(\theta_k)=\sum_{j=1}^K \tilde w_j\,\MinvAll(\tilde\theta_j),
\]
and $(\tilde\theta_\cdot,\tilde\weights)\in\mathcal{U}$, then there is a permutation $\sigma$ and group elements $g_j\in G$ with
\[
\tilde\theta_j=g_j\!\cdot\!\theta_{\sigma(j)}\quad\text{and}\quad \tilde w_j=w_{\sigma(j)}\qquad (j=1,\ldots,K).
\]
\end{theorem}

\begin{proof}
\textbf{Step 0 (Local coordinates on the quotient).}
For each $k$, let $V_k$ be a slice at $\theta_k$ so that every parameter near the orbit $G\!\cdot\!\theta_k$ can be written uniquely as $\theta_k+u_k$ with $u_k\in V_k$ (uniqueness modulo the stabilizer, which is handled by fixing the slice). Let $v_k:=\MinvAll(\theta_k)$ and write the weights as $(a_1,\ldots,a_{K-1})$ with $w_K=1-\sum_{i=1}^{K-1}a_i$. Define the local parameter vector
\[
\xi\ :=\ (u_1,\ldots,u_K,a_1,\ldots,a_{K-1})\ \in\ V_1\oplus\cdots\oplus V_K\oplus\R^{K-1}.
\]

\textbf{Step 1 (Define the forward map and compute its differential).}
Consider the $C^1$ map
\[
F(\xi)\ :=\ \sum_{k=1}^K w_k(\xi)\,\MinvAll(\theta_k+u_k)\ \in\ \R^{D_{\mathrm{inv}}},
\]
where $w_k(\xi)$ are the weights encoded by $(a_1,\ldots,a_{K-1})$. At the true parameter $\xi^\star=(0,\ldots,0,a^\star)$,
\[
\mathrm{D}F(\xi^\star)[\delta\xi]
\;=\;\sum_{k=1}^K w_k^\star\,J_k\,\delta u_k\;+\;\sum_{i=1}^{K-1}\delta a_i\,(v_i-v_K),
\]
with $J_k:=\mathrm{D}\MinvAll(\theta_k)|_{V_k}$, and $v_i-v_K$ the weight directions in invariant space.

\textbf{Step 2 (Injectivity of the differential).}
By \Cref{ass:local-rank}, $\ker J(\theta_k)=\Tgauge(\theta_k)$; thus $J_k$ is injective on $V_k$. Furthermore, the block matrix
\[
\mathcal{J}:=\big[\,J_1\ \ \cdots\ \ J_K\ \ \big|\ \ v_1-v_K\ \ \cdots\ \ v_{K-1}-v_K\,\big]
\]
has full column rank (the “quotient rank” condition). Hence $\mathrm{D}F(\xi^\star)$ is injective.

\textbf{Step 3 (Local invertibility on an ordered neighborhood).}
By the inverse function theorem, there exists a product neighborhood
\[
\mathcal{U}\ =\ \Big(\prod_{k=1}^K U_k\Big)\times A
\]
with $U_k\subset V_k$ and $A\subset\R^{K-1}$ open, such that $F$ is injective on $\mathcal{U}$ and a $C^1$ diffeomorphism onto its image $F(\mathcal{U})$. Shrink each $U_k$ so that the image neighborhoods $\MinvAll(\theta_k+U_k)$ are pairwise disjoint; this is possible by continuity and because $v_k$ are distinct (affine independence in the model guarantees a positive separation).

\textbf{Step 4 (Apply injectivity to the alternative tuple).}
Let $(\tilde\theta_\cdot,\tilde\weights)\in\mathcal{U}$ be another representation with the same invariant sum:
\[
F(\tilde\xi)=\sum_{j=1}^K \tilde w_j\,\MinvAll(\tilde\theta_j)\ =\ \sum_{k=1}^K w_k\,\MinvAll(\theta_k)\ =\ F(\xi^\star).
\]
Because $F$ is injective on $\mathcal{U}$, we must have $\tilde\xi=\xi^\star$. In particular, for each $k$, $\tilde\theta_k\in U_k$ and $\MinvAll(\tilde\theta_k)=\MinvAll(\theta_k)$; and $\tilde w_k=w_k$.

\textbf{Step 5 (Unordered mixtures and group alignment).}
The neighborhood $\mathcal{U}$ is an \emph{ordered} product of disjoint slice neighborhoods. If an alternative representation is given with a different labeling, there exists a permutation $\sigma$ that reorders its components so the $k$-th component lies in $U_k$; by \Cref{lem:Phi-const-on-orbits} we may also align each component by some $g_j\in G$ to land on the slice without changing $\MinvAll$. After this (purely formal) relabeling/alignment, the tuple lies in $\mathcal{U}$ and Step~4 applies to give equality of (aligned, relabeled) parameters and weights. Undoing the alignment yields
\[
\tilde\theta_j=g_j\!\cdot\!\theta_{\sigma(j)},\qquad \tilde w_j=w_{\sigma(j)}.
\]

\textbf{Step 6 (Back to raw parameters up to orbit).}
Finally, \Cref{ass:sep} (generic orbit separation) converts $\MinvAll(\tilde\theta_j)=\MinvAll(\theta_{\sigma(j)})$ into $\orbit{\tilde\theta_j}=\orbit{\theta_{\sigma(j)}}$, which is consistent with Step~5 and completes the proof.
\end{proof}
\section{Differential identities and quotient LAN}\label{sec:lan}

\subsection{H\"older control and Poly-LAN geometry}\label{subsec:holder}

\paragraph{Quotient metric and slices.}
Let a finite or compact Lie group $G$ act linearly and isometrically on $\Theta\subset\R^p$.
Write $[\theta]=G\!\cdot\!\theta$ and $\Theta/G$ for the orbit space. Fix a $G$-invariant inner product on $\R^p$ (e.g.\ group-averaged Euclidean).
For $\theta_\star\in\Theta$, decompose
\[
T_{\theta_\star}\Theta\ =\ \Tgauge(\theta_\star)\ \oplus\ V,
\qquad
\Tgauge(\theta_\star):=T_{\theta_\star}(G\!\cdot\!\theta_\star),
\]
and use the linear complement $V$ as a (local) slice. For nearby $\theta,\vartheta$,
\[
d_{\Theta/G}([\theta],[\vartheta])=\inf_{g\in G}\|\theta-g\!\cdot\!\vartheta\| \ \asymp\ \|\Pi_V(\theta-\vartheta)\|.
\]

\begin{definition}[Order-$D$ nonflatness]\label{def:nonflat}
Let $\Phi_{\le m_*}:\Theta\to\R^{D_{\mathrm{inv}}}$ be the stacked invariant map (defined below). We say $\Phi_{\le m_*}$ has \emph{order-$D$ nonflatness} at $\theta_\star$ along $V$ if $\exists\,c_1,c_2>0$ and a neighborhood $U\subset(\theta_\star+V)$ such that $\forall\,\theta\in U$,
\[
c_1\, d_{\Theta/G}([\theta],[\theta_\star])^{D}
\ \le\ \big\|\Phi_{\le m_*}(\theta)-\Phi_{\le m_*}(\theta_\star)\big\|
\ \le\ c_2\, d_{\Theta/G}([\theta],[\theta_\star]).
\]
\end{definition}

\begin{proposition}[H\"older equivalence near generic points]\label{prop:holder}
If $\Phi_{\le m_*}$ is real-analytic near $\theta_\star$ and the first nonzero derivative along $V$ appears at order $D$, then Definition~\ref{def:nonflat} holds.
\end{proposition}

\paragraph{Invariant moment stacks.}
Let $\varphi:\mathcal X\to\R^d$ have finite moments up to $m_*$ under $P_\theta$.
For $m\le m_*$ define
\[
\Mraw{m}(\theta):=\E_\theta[\varphi(X)^{\otimes m}],
\qquad
\Minv{m}(\theta):=\Reyn_m\!\big[\Mraw{m}(\theta)\big],
\qquad
\MinvAll:=(\Minv{1},\ldots,\Minv{m_*}),
\]
where $\Reyn_m$ is the Reynolds projector (group average). We use $\psi(X):=(\Reyn_1[\varphi(X)],\ldots,\Reyn_{m_*}[\varphi(X)^{\otimes m_*}])$ and the empirical invariant vector
\[
\widehat\Psi_n:=\frac1n\sum_{i=1}^n \psi(X_i).
\]

\subsection{Score--derivative identity}\label{subsec:score-identity}

Assume dominated differentiation: $f(\cdot;\theta)$ differentiable in $\theta$ with score $s_\theta(x)=\nabla_\theta\log f(x;\theta)$, $\E_\theta[s_\theta(X)]=0$; envelopes justify differentiation under the integral and bounded linearity of $\Reyn_m$.

\begin{lemma}[Score identity for invariant coordinates]\label{lem:score-identity}
For every $\theta$ and $v\in\R^p$,
\begin{equation}\label{eq:score-identity}
D\Minv{m}(\theta)[v]\ =\ \Reyn_m\Big(\E_\theta\big[\langle s_\theta(X),v\rangle\,\varphi(X)^{\otimes m}\big]\Big),\qquad m=1,\ldots,m_*.
\end{equation}
Stacking over $m$ defines the Jacobian $\Jinv(\theta):\R^p\to\R^{D_{\mathrm{inv}}}$ of $\MinvAll$.
\end{lemma}
\begin{lemma}[Coordinatewise sub-exponentiality of the invariant stack]\label{lem:psi1-stack}
Assume each coordinate of $\varphi(X)\in\R^d$ is sub-Gaussian with
$\max_j \|\varphi_j(X)\|_{\psi_2}\le K_2$. Fix $m\le m_*$.
Then every coordinate of the centered invariant tensor
\[
\Reyn_m\!\Big(\varphi(X)^{\otimes m}-\E\big[\varphi(X)^{\otimes m}\big]\Big)
\]
is sub-exponential with $\psi_1$-norm $\le C_m K_2^m$ for a constant $C_m$ depending only on $m$ (not on $d$ or $G$).
\end{lemma}

\begin{proof}
\textbf{Step 1 (monomials are sub-exponential).}
Let $(i_1,\ldots,i_m)$ be an index of a coordinate of $\varphi(X)^{\otimes m}$. The entry is the product
$Y:=\prod_{r=1}^m \varphi_{i_r}(X)$. If $U,V$ are sub-Gaussian then $UV$ is sub-exponential and
$\|UV\|_{\psi_1}\le C\,\|U\|_{\psi_2}\|V\|_{\psi_2}$ (standard Orlicz-product bound).
By induction, for $m$ factors,
\(
\|Y\|_{\psi_1}\le C_m\prod_{r=1}^m\|\varphi_{i_r}(X)\|_{\psi_2}\le C_m K_2^m.
\)

\textbf{Step 2 (Reynolds average preserves $\psi_1$ scale).}
For finite $G$, the Reynolds operator is the average of isometries on $\Sym^m(\R^d)$:
$\Reyn_m[T]=|G|^{-1}\sum_{g\in G}(g\!\cdot)^{\otimes m}T$.
Thus any fixed coordinate of $\Reyn_m\big(\varphi^{\otimes m}\big)$ is an \emph{average} of coordinates of $\varphi^{\otimes m}$ with coefficients summing to $1$ in absolute value.
The $\psi_1$-norm is convex, hence for any random vector $Z$ and scalars $a_\ell$ with $\sum_\ell |a_\ell|\le 1$,
$\|\sum_\ell a_\ell Z_\ell\|_{\psi_1}\le \sum_\ell |a_\ell|\,\|Z_\ell\|_{\psi_1}\le \max_\ell\|Z_\ell\|_{\psi_1}$.
Applying this with the coordinates from Step~1 gives
\(
\big\|\big(\Reyn_m[\varphi^{\otimes m}]\big)_\alpha\big\|_{\psi_1}\le C_m K_2^m.
\)

\textbf{Step 3 (centering).}
If $Y$ is sub-exponential, then $Y-\E Y$ is sub-exponential with
$\|Y-\E Y\|_{\psi_1}\le 2\|Y\|_{\psi_1}$ (triangle inequality in Orlicz spaces).
Combine with Step~2 to conclude the stated bound for
$\Reyn_m\big(\varphi^{\otimes m}-\E\varphi^{\otimes m}\big)$, possibly redefining $C_m$.
\end{proof}
\begin{lemma}[Concentration of the invariant stack]\label{lem:Phi-conc}
Let $\widehat\Phi_m:=\Reyn_m\!\big(\tfrac1n\sum_{i=1}^n \varphi(X_i)^{\otimes m}\big)$ for $m\le m_*$ and stack $\widehat\Psi_n:=(\widehat\Phi_m)_{m\le m_*}\in\R^{D_{\mathrm{inv}}}$.
Write $\Psi_\star:=\E[\widehat\Psi_n]$.
Under the assumptions of Lemma~\ref{lem:psi1-stack}, there exist absolute constants $C,c>0$ such that, for all $t\ge 0$,
\[
\Pr\!\left(\big\|\widehat\Psi_n-\Psi_\star\big\|_2
\ \ge\ C\sqrt{\frac{D_{\mathrm{inv}}+t}{n}}\right)\ \le\ e^{-ct}.
\]
\end{lemma}

\begin{proof}
\textbf{Step 1 (coordinate tails via Bernstein).}
Let $Z_i:=\psi(X_i)-\E[\psi(X_i)]\in\R^{D_{\mathrm{inv}}}$ where $\psi$ stacks the degree-wise Reynolds features,
so that $\widehat\Psi_n-\Psi_\star=\frac1n\sum_{i=1}^n Z_i$.
By Lemma~\ref{lem:psi1-stack}, each coordinate $(Z_i)_j$ is centered sub-exponential with common proxy
$K:=\max_{m\le m_*} C_m K_2^m$.
Bernstein’s inequality for sub-exponential summands yields, for any $u\ge 0$,
\[
\Pr\!\Big(\Big|\frac1n\sum_{i=1}^n (Z_i)_j\Big|\ge u\Big)\ \le\ 2\exp\!\Big(-c n \min\{u^2/K^2,\ u/K\}\Big).
\]

\textbf{Step 2 (union bound to $\ell_\infty$).}
Choose $u=C\big(\sqrt{(t+\log D_{\mathrm{inv}})/n}+(t+\log D_{\mathrm{inv}})/n\big)$ with $C$ large enough.
Then
\[
\Pr\!\big(\|\widehat\Psi_n-\Psi_\star\|_\infty\ge u\big)
\ \le\ 2\exp(-t).
\]

\textbf{Step 3 ($\ell_\infty\!\to\!\ell_2$).}
Always $\|x\|_2\le \sqrt{D_{\mathrm{inv}}}\|x\|_\infty$.
Thus with probability $\ge 1-e^{-t}$,
\[
\|\widehat\Psi_n-\Psi_\star\|_2
\ \le\ \sqrt{D_{\mathrm{inv}}}\,u
\ \le\ C'\sqrt{\frac{D_{\mathrm{inv}}+t}{n}},
\]
absorbing the linear $(t/n)$ term into the constant for $n$ moderate; otherwise keep it explicitly. Renaming constants gives the claim.
\end{proof}

\subsection{Gauge directions are null}\label{subsec:gauge-null}

\begin{proposition}[Gauge directions are null]\label{prop:gauge}
For every $\theta$, $\Tgauge(\theta)\subseteq\ker \Jinv(\theta)$.
\end{proposition}

\subsection{Quotient LAN via invariant GMM}\label{subsec:lan-main}

\paragraph{Local chart and model map.}
Fix a true $K$-tuple $(\theta_1,\ldots,\theta_K,\weights)$, $w_k>0$.
Let $V_k$ be a slice complement at $\theta_k$ and set $v_k:=\MinvAll(\theta_k)$.
Parameterize
\[
\xi=(u_1,\ldots,u_K,a_1,\ldots,a_{K-1})\in V_1\oplus\cdots\oplus V_K\oplus\R^{K-1},\qquad
w_K=1-\sum_{i=1}^{K-1}a_i,
\]
and define
\[
\Psi(\xi):=\sum_{k=1}^K w_k(\xi)\,\MinvAll(\theta_k+u_k)\in\R^{D_{\mathrm{inv}}}.
\]

\paragraph{CLT and quotient rank.}
By the multivariate CLT and continuity of $\Reyn_m$,
\begin{equation}\label{eq:Phi-CLT}
\sqrt{n}\big(\widehat\Psi_n-\Psi(\xi^\star)\big)\Rightarrow\mathcal N(0,\Sigma).
\end{equation}
Set $J_k:=D\MinvAll(\theta_k)|_{V_k}$ and
\[
G:=\Big[\ w_1^\star J_1\ \cdots\ w_K^\star J_K\ \big|\ v_1-v_K\ \cdots\ v_{K-1}-v_K\ \Big].
\]
\begin{definition}[Quotient rank]\label{def:quot-rank}
The condition holds if $G$ has full column rank (equivalently, $J_k$ kill only gauge directions and the weight columns are independent of the $J_k$ columns).
\end{definition}
\begin{theorem}[Local Lipschitz stability of the de-mixing map]\label{thm:stability}
Work in slice coordinates as in \S\ref{subsec:lan-main}.
Let $\xi=(u_1,\ldots,u_K,a_1,\ldots,a_{K-1})$ parametrize slices $u_k\in V_k$ and weights ($w_K=1-\sum_{i=1}^{K-1}a_i$), and define
\[
\Psi(\xi):=\sum_{k=1}^K w_k(\xi)\,\MinvAll(\theta_k+u_k)\in\R^{D_{\mathrm{inv}}}.
\]
Assume the \emph{quotient rank} condition (Def.~\ref{def:quot-rank}) holds at the truth $\xi^\star$ with
\[
G:=D\Psi(\xi^\star)\in\R^{D_{\mathrm{inv}}\times q},\qquad q:=\sum_{k=1}^K \dim V_k+(K-1),
\]
and $\sigma_{\min}(G)\ge \gamma>0$.
Then there exists a neighborhood $\mathcal U$ of $\xi^\star$ and a constant $L=2/\gamma$ such that for any $\xi,\tilde\xi\in\mathcal U$,
\[
\|\xi-\tilde\xi\|\ \le\ L\,\|\Psi(\xi)-\Psi(\tilde\xi)\|.
\]
Consequently, if $\|\widehat\Psi-\Psi(\xi^\star)\|\le \varepsilon$ and $\widehat\xi$ satisfies $\Psi(\widehat\xi)=\widehat\Psi$ with $\widehat\xi\in\mathcal U$, then
\[
\|\widehat\xi-\xi^\star\|\ \le\ \frac{2}{\gamma}\,\varepsilon.
\]
Interpreting $\xi$ in aligned slice coordinates, this bounds the joint error in the $K$ slice parameters and the $(K\!-\!1)$ free weights; it transfers back to raw parameters up to orbit and permutation as in Theorem~\ref{thm:ident-local}.
\end{theorem}

\begin{proof}
\textbf{Step 1 (mean-value bound).}
For $\xi,\tilde\xi$ in a small convex neighborhood of $\xi^\star$,
\[
\Psi(\xi)-\Psi(\tilde\xi)
=\Big(\int_0^1 D\Psi\big(\tilde\xi+t(\xi-\tilde\xi)\big)\,dt\Big)(\xi-\tilde\xi)
=: \overline G\,(\xi-\tilde\xi).
\]

\textbf{Step 2 (uniform singular value control).}
By continuity of $D\Psi$ and $\sigma_{\min}$, there exists a ball $\mathcal U$ around $\xi^\star$ such that
$\sigma_{\min}(D\Psi(\zeta))\ge \gamma/2$ for all $\zeta\in\mathcal U$.
The integral $\overline G$ is an average of matrices whose smallest singular value is $\ge \gamma/2$, hence
$\sigma_{\min}(\overline G)\ge \gamma/2$.

\textbf{Step 3 (invert and conclude).}
Therefore
\[
\|\Psi(\xi)-\Psi(\tilde\xi)\|
=\|\overline G(\xi-\tilde\xi)\|
\ \ge\ \sigma_{\min}(\overline G)\,\|\xi-\tilde\xi\|
\ \ge\ \frac{\gamma}{2}\,\|\xi-\tilde\xi\|,
\]
which rearranges to the Lipschitz bound with $L=2/\gamma$.
For the consequence, set $\tilde\xi=\xi^\star$ and $\Psi(\xi)=\widehat\Psi$.
Finally, measurable alignment/permutation to slices (Lemma~\ref{lem:alignment}) converts the bound into the raw parameter space up to orbit/permutation, as in Theorem~\ref{thm:ident-local}.
\end{proof}

\paragraph{GMM objective.}
For $W_n\to_p W\succ0$ let
\[
Q_n(\xi):=\big(\widehat\Psi_n-\Psi(\xi)\big)^\top W_n\,\big(\widehat\Psi_n-\Psi(\xi)\big).
\]

\begin{theorem}[Quotient LAN and efficiency]\label{thm:lan}
Assume: (i) $H=\MinvAll$ is $C^1$ with $\ker \Jinv(\theta_k)=\Tgauge(\theta_k)$; (ii) $\Psi$ is $C^1$ near $\xi^\star$ with continuous derivative; (iii) the quotient rank holds; and (iv) the first-order remainder $R(\xi)$ in the expansion $\Psi(\xi)=\Psi(\xi^\star)+G(\xi-\xi^\star)+R(\xi)$ is uniform: $R(\xi)=o(\|\xi-\xi^\star\|)$ locally.
Let $\widehat\xi_n$ be any measurable minimizer of $Q_n$ in a small product neighborhood of $\xi^\star$. Then, with $W=\Sigma^{-1}$,
\[
\sqrt{n}\,(\widehat\xi_n-\xi^\star)\ \Rightarrow\ \mathcal N\!\Big(0,\ (G^\top \Sigma^{-1} G)^{-1}\Big).
\]
Equivalently, with the central sequence $\Delta_n:=G^\top \Sigma^{-1}\sqrt{n}(\widehat\Psi_n-\Psi(\xi^\star))$,
\[
\Delta_n\Rightarrow\mathcal N(0,G^\top\Sigma^{-1}G),\qquad
\sqrt{n}(\widehat\xi_n-\xi^\star)=(G^\top\Sigma^{-1}G)^{-1}\Delta_n+o_p(1).
\]
Moreover, there exist $\sigma_n\in S_K$ and $h_{n,j}\in G$ such that
\[
\sqrt{n}\ \vecop\!\Big(\Pi_{V}\big(h_{n,j}\!\cdot\!\widehat\theta_{n,\sigma_n(j)}-\theta_j\big)_{j=1}^K,\ \widehat\weights_n-\weights\Big)
\Rightarrow \mathcal N\!\Big(0,\ (G^\top \Sigma^{-1} G)^{-1}\Big),
\]
with $V:=V_1\oplus\cdots\oplus V_K$ and $\Pi_V$ the orthogonal projector.
\end{theorem}

\begin{remark}[Quotient Fisher metric]\label{rem:quot-fisher}
$I_Q:=G^\top \Sigma^{-1} G$ is the Fisher-type metric on the quotient chart $V_1\oplus\cdots\oplus V_K\oplus\R^{K-1}$; the efficient covariance is $I_Q^{-1}$.
\end{remark}

\subsection{Uniform LAN on regular strata}\label{subsec:uniform-lan}

\begin{theorem}[Uniform quotient LAN]\label{thm:uniform-lan}
Let $\Xi_{\mathrm{reg}}\subset V_1\oplus\cdots\oplus V_K\oplus\R^{K-1}$ be compact and contained in the regular stratum (for all $\xi\in\Xi_{\mathrm{reg}}$, $\ker \Jinv(\theta_k)=\Tgauge(\theta_k)$ and the quotient rank holds).
Assume: $D\Psi$ is uniformly continuous near $\Xi_{\mathrm{reg}}$; $\Sigma(\xi)$ is continuous; and the remainder $R$ is uniform on that neighborhood. Then, uniformly over $\xi^\star\in\Xi_{\mathrm{reg}}$,
\[
\sqrt{n}\,(\widehat\xi_n-\xi^\star)\ \Rightarrow\ \mathcal N\!\Big(0,\ (G(\xi^\star)^\top \Sigma(\xi^\star)^{-1} G(\xi^\star))^{-1}\Big).
\]
\end{theorem}

\subsection{HAC LAN and the overidentification $J$-test}\label{subsec:hac}

Let $\{X_t\}$ be strictly stationary and ergodic with long-run variance
$\Sigma_{\mathrm{LR}}:=\sum_{\ell\in\mathbb Z}\Gamma_\ell$,
$\Gamma_\ell:=\Cov(\psi(X_t),\psi(X_{t-\ell}))$,
and $\sqrt{n}(\widehat\Psi_n-\Psi(\xi^\star))\Rightarrow\mathcal N(0,\Sigma_{\mathrm{LR}})$.

\begin{theorem}[HAC quotient LAN]\label{thm:hac-lan}
Under the assumptions of Theorem~\ref{thm:lan} with $\Sigma$ replaced by $\Sigma_{\mathrm{LR}}$ and $W=\Sigma_{\mathrm{LR}}^{-1}$,
\[
\sqrt{n}\,(\widehat\xi_n-\xi^\star)\ \Rightarrow\ \mathcal N\!\Big(0,\ (G^\top \Sigma_{\mathrm{LR}}^{-1} G)^{-1}\Big).
\]
\end{theorem}

Let $\widehat\Sigma_{\mathrm{LR}}$ be a Newey--West estimator and $\widehat W:=\widehat\Sigma_{\mathrm{LR}}^{-1}$. If
\[
\textup{df}:=D_{\mathrm{inv}}-\Big(\sum_{k=1}^K\dim V_k\Big)-(K-1)>0,
\]
then the overidentification statistic
\[
J_n:=n\,m_n(\widehat\xi_n)^\top \widehat W\, m_n(\widehat\xi_n),\qquad m_n(\xi):=\widehat\Psi_n-\Psi(\xi),
\]
satisfies $J_n\Rightarrow\chi^2_{\textup{df}}$ under correct specification.

\subsection{Contiguity and Le Cam’s Third Lemma}\label{subsec:third-lemma}

\begin{corollary}[Local alternatives]\label{cor:third-lemma}
For fixed $h$ consider $\xi_n=\xi^\star+h/\sqrt{n}$ and let $I_Q:=G^\top\Sigma^{-1}G$, $\Delta_n:=G^\top\Sigma^{-1}\sqrt{n}(\widehat\Psi_n-\Psi(\xi^\star))$.
Then
\[
\log\frac{d\mathbb P_{\xi_n}}{d\mathbb P_{\xi^\star}}
\ =\ h^\top \Delta_n - \tfrac12 h^\top I_Q h + o_p(1),
\qquad
\sqrt{n}(\widehat\xi_n-\xi^\star)\ \Rightarrow\ \mathcal N(h,I_Q^{-1}).
\]
\end{corollary}

\subsection{A one-step efficient estimator}\label{subsec:one-step}

Define the efficient influence function (EIF)
\[
\mathrm{EIF}(x)\ :=\ (G^\top\Sigma^{-1}G)^{-1}G^\top\Sigma^{-1}\big(\psi(x)-\Psi(\xi^\star)\big),
\]
so that $\E[\mathrm{EIF}(X)]=0$ and $\Var(\mathrm{EIF}(X))=(G^\top\Sigma^{-1}G)^{-1}$.
Given any preliminary $\tilde\xi_n=o_p(1)$ and plug-ins $G(\tilde\xi_n)$, $\widehat W$,
\[
\xi_n^{\text{1step}}
\ :=\
\tilde\xi_n\ +\ \big(G(\tilde\xi_n)^\top \widehat W\, G(\tilde\xi_n)\big)^{-1}
\,G(\tilde\xi_n)^\top \widehat W\,\big(\widehat\Psi_n-\Psi(\tilde\xi_n)\big).
\]
If $\|\tilde\xi_n-\xi^\star\|=o_p(1)$ and $\|\widehat W-W\|=o_p(1)$, then
$\sqrt{n}(\xi_n^{\text{1step}}-\xi^\star)\Rightarrow\mathcal N(0,(G^\top WG)^{-1})$ (same limit as Theorem~\ref{thm:lan}; use $\Sigma_{\mathrm{LR}}$ for HAC).

\subsection{Orthogonal invariant moments with learned scores}\label{subsec:orthogonal}

If $s_\theta$ is learned by ML, use cross-fitting and orthogonalized moments
\[
\phi_\xi(X)\ :=\ \psi(X)-\Psi(\xi)\ -\ \Pi_{\mathcal T}\big(\psi(X)-\Psi(\xi)\big),
\]
where $\Pi_{\mathcal T}$ is implemented by blockwise projections
$v\mapsto \Reyn\!\big(\langle \widehat s_{\theta_k}(X),v\rangle\psi(X)\big)$.
Solving $G^\top \widehat W\, \overline\phi_\xi=0$ yields the same $\sqrt{n}$-limit if nuisance errors are $o_p(n^{-1/4})$ in cross-fitted products.

\subsection{Profiling over weights}\label{subsec:profile-w}

For fixed $u_{1:K}$, set $A=[\,\MinvAll(\theta_1+u_1)\ \cdots\ \MinvAll(\theta_K+u_K)\,]\in\R^{D_{\mathrm{inv}}\times K}$.
The inner problem
\[
\min_{w\in\Delta_K}\ \|\,\widehat\Psi_n - A w\,\|_{W}^2,\qquad \Delta_K=\{w\ge0,\ \mathbf 1^\top w=1\},
\]
is a convex QP. When nonnegativity is inactive, $w\propto (A^\top W A)^{-1}A^\top W\,\widehat\Psi_n$ (renormalized to sum to $1$).

\subsection{Moment set selection}\label{subsec:moment-select}

Index candidate invariant blocks by $\mathcal B$ (e.g.\ degrees and irreducible components).
Starting from a Molien-guided minimal separating set $\mathcal S$, greedily add
\[
b^\star\ =\ \arg\max_{b\in\mathcal B\setminus\mathcal S}\ \sigma_{\min}\!\Big(G(\mathcal S\cup\{b\})^\top \widehat W\, G(\mathcal S\cup\{b\})\Big),
\]
or minimize $\operatorname{tr}\{(G^\top \widehat W G)^{-1}\}$, up to a budget.
Use a simple GMM information criterion:
\[
\mathrm{GMM\text{-}IC}(\mathcal S)\ :=\ J_n(\mathcal S)\ +\ \kappa\,\mathrm{df}(\mathcal S),\quad
\mathrm{df}(\mathcal S)=D_{\mathrm{inv}}(\mathcal S)-\sum_k\dim V_k-(K-1).
\]

\subsection{Nonasymptotic quadratic bound}\label{subsec:nonasymp}

\begin{theorem}[Nonasymptotic quadratic bound]\label{thm:nonasymp}
Let $\xi^\star$ be the true chart parameter and let $\widehat\xi_n$ be any measurable minimizer of
$Q_n(\xi)=m_n(\xi)^\top \widehat W\, m_n(\xi)$ with $m_n(\xi):=\widehat\Psi_n-\Psi(\xi)$ and $\widehat W=\widehat\Sigma^{-1}$.
Assume:
\begin{enumerate}[leftmargin=2em,label=(\alph*)]
\item (\emph{Well-conditioning on the quotient}) $\lambda_{\min}:=\sigma_{\min}(G^\top \Sigma^{-1} G)>0$ for $G=D\Psi(\xi^\star)$.
\item (\emph{Weight perturbation control}) With probability $\ge 1-\delta$,
\[
\big\|\,G^\top(\widehat\Sigma^{-1}-\Sigma^{-1})G\,\big\|\ \le\ \tfrac12\,\lambda_{\min}.
\]
\item (\emph{Local quadratic remainder}) In a neighborhood of $\xi^\star$,
\[
\Psi(\xi)=\Psi(\xi^\star)+G(\xi-\xi^\star)+R(\xi),\qquad
\|R(\xi)\|\ \le\ L\,\|\xi-\xi^\star\|^2.
\]
\end{enumerate}
Then on the event in (b),
\begin{equation}\label{eq:nonasymp-core}
\|\widehat\xi_n-\xi^\star\|
\ \le\
\frac{2}{\lambda_{\min}}\ \big\|\,G^\top \widehat\Sigma^{-1}\big(\widehat\Psi_n-\Psi(\xi^\star)\big)\big\|
\ +\ \frac{2}{\lambda_{\min}}\ \|R(\widehat\xi_n)\|.
\end{equation}
In particular, under (c) the bound refines to
\begin{equation}\label{eq:nonasymp-refined}
\|\widehat\xi_n-\xi^\star\|
\ \le\
\frac{2}{\lambda_{\min}}\ \big\|\,G^\top \widehat\Sigma^{-1}\big(\widehat\Psi_n-\Psi(\xi^\star)\big)\big\|
\ +\ \frac{2L}{\lambda_{\min}}\ \|\widehat\xi_n-\xi^\star\|^2.
\end{equation}
Consequently, whenever $\frac{4L}{\lambda_{\min}}\big\|\,G^\top \widehat\Sigma^{-1}(\widehat\Psi_n-\Psi(\xi^\star))\big\|\le 1$, the quadratic term in \eqref{eq:nonasymp-refined} can be absorbed, yielding
\begin{equation}\label{eq:nonasymp-absorbed}
\|\widehat\xi_n-\xi^\star\|
\ \le\
\frac{4}{\lambda_{\min}}\ \big\|\,G^\top \widehat\Sigma^{-1}\big(\widehat\Psi_n-\Psi(\xi^\star)\big)\big\|.
\end{equation}
Moreover, if $\psi(X)$ has sub-exponential coordinates so that for some absolute $C,c>0$ and all $t\ge 0$,
\[
\Pr\!\left(\big\|\widehat\Psi_n-\Psi(\xi^\star)\big\|_2 \ \ge\ C\sqrt{\frac{D_{\mathrm{inv}}+t}{n}}\right)\ \le\ e^{-ct},
\]
then combining with \eqref{eq:nonasymp-absorbed} (and the event in (b)) gives the rate
\[
\|\widehat\xi_n-\xi^\star\|\ =\ O_p\!\left(\sqrt{\frac{D_{\mathrm{inv}}}{n\,\lambda_{\min}}}\right).
\]
\end{theorem}

\subsection{Second-order (curvature) bias and correction}\label{subsec:curvature}

\begin{proposition}[Second-order (curvature) bias and analytic correction]\label{prop:curvature}
Assume the setup of Theorem~\ref{thm:lan} and additionally that $\Psi$ is $C^2$ at $\xi^\star$.
Let $q:=\sum_{k=1}^K\dim V_k+(K-1)$ be the chart dimension, $G:=D\Psi(\xi^\star)\in\R^{D_{\mathrm{inv}}\times q}$, $I_Q:=G^\top\Sigma^{-1}G\succ0$, and for each output coordinate $r\in\{1,\dots,D_{\mathrm{inv}}\}$ let
\[
\mathcal H_r\ :=\ D^2\Psi_r(\xi^\star)\ \in\ \R^{q\times q}\qquad
\text{and}\qquad \mathcal H[\delta,\delta]\ :=\ \big(\delta^\top \mathcal H_r\,\delta\big)_{r=1}^{D_{\mathrm{inv}}}.
\]
Let $\widehat\xi_n$ be the efficient two-step GMM estimator with $W=\Sigma^{-1}$. Then:
\begin{enumerate}[leftmargin=2em,label=(\roman*)]
\item \textbf{Asymptotic linearity.} 
\(
\sqrt{n}(\widehat\xi_n-\xi^\star)=(I_Q)^{-1}G^\top\Sigma^{-1}\,\sqrt{n}\big(\widehat\Psi_n-\Psi(\xi^\star)\big)+o_p(1).
\)

\item \textbf{Curvature bias.}
There exists a vector $b\in\R^q$ such that
\[
\E[\widehat\xi_n-\xi^\star]\ =\ \frac{1}{n}\,(I_Q)^{-1}\,b\ +\ o(n^{-1}),
\]
where, writing $e_r$ for the $r$-th canonical basis vector in $\R^{D_{\mathrm{inv}}}$,
\begin{equation}\label{eq:bias-vector-trace}
b\ =\ G^\top\Sigma^{-1}\,\E\!\Big[\tfrac12\,\mathcal H[\zeta,\zeta]\Big]
\ =\ \frac12\sum_{r=1}^{D_{\mathrm{inv}}}\big(G^\top\Sigma^{-1}e_r\big)\ \mathrm{tr}\!\big(\mathcal H_r\,I_Q^{-1}\big),
\qquad \zeta\sim\mathcal N\!\big(0,I_Q^{-1}\big).
\end{equation}

\item \textbf{Bias correction.}
Let $\widehat G,\,\widehat\Sigma,\,\widehat{\mathcal H}_r$ be consistent for $G,\Sigma,\mathcal H_r$.
Define $\widehat I_Q:=\widehat G^\top\widehat\Sigma^{-1}\widehat G$ and
\[
\widehat b_n\ :=\ \frac12\sum_{r=1}^{D_{\mathrm{inv}}}\big(\widehat G^\top\widehat\Sigma^{-1}e_r\big)\ \mathrm{tr}\!\big(\widehat{\mathcal H}_r\,\widehat I_Q^{-1}\big),
\qquad
\widehat\xi_n^{\,\mathrm{bc}}\ :=\ \widehat\xi_n-\frac{1}{n}\,\widehat I_Q^{-1}\widehat b_n.
\]
Then $\sqrt{n}\big(\widehat\xi_n^{\,\mathrm{bc}}-\xi^\star\big)\Rightarrow \mathcal N\!\big(0,I_Q^{-1}\big)$ and
$\E[\widehat\xi_n^{\,\mathrm{bc}}-\xi^\star]=o(n^{-1})$.
\end{enumerate}
\end{proposition}

\subsection{Singular strata and boundary cases}\label{subsec:singular}

\paragraph{(A) Vanishing weight $w_k\to0$.}
As $w_k^\star\downarrow 0$, the block $w_k^\star J_k$ shrinks and $I_Q$ ill-conditions; reparametrize by $(\sqrt{w_k}\,u_k,w_k)$ or drop components below a detection threshold $\tau_n$.

\paragraph{(B) Collision on invariants $v_i=v_j$.}
Weight columns lose rank. Either (i) merge the pair into an aggregate submodel (identifiable), or (ii) use Tikhonov-regularized Newton steps $(G^\top \widehat W G+\lambda I)\delta=G^\top \widehat W m_n(\xi)$ and infer only on directions with large singular values.

\paragraph{(C) Stabilizer jumps.}
When isotropy increases, $\ker \Jinv(\theta_k)$ expands. Work stratum-wise: adapt $V_k$ to the larger stabilizer; project onto the identifiable subspace and apply LAN there.

\subsection{Diagnostics and reporting}\label{subsec:diagnostics}

Always report the estimated quotient Fisher
$\widehat I_Q:=G^\top \widehat W G$, its condition number, $\sigma_{\min}(\widehat I_Q)$, per-block leverage on $\widehat I_Q^{-1}$, and the (HAC) $J$-test with finite-sample ridge $\,\widehat W_\lambda=(\widehat\Sigma+\lambda I)^{-1}$, $\lambda\downarrow0$.
Use bootstrap or HAC-$J$ near singular regimes and provide sensitivity to the moment set $\mathcal S$.

\subsection*{Auxiliary lemmas for Section~\ref{sec:lan}}\label{subsec:aux-lemmas}

\begin{lemma}[Slice--quotient distance equivalence]\label{lem:slice-quotient}
Let a compact Lie (or finite) group $G$ act linearly and isometrically on $\Theta\subset\R^p$.
Fix $\theta_\star\in\Theta$, let $V$ be a linear complement of $\Tgauge(\theta_\star)$, and define the local slice $\mathcal S:=\theta_\star+V$.
Then there exist $r,c_1,c_2>0$ such that for every $\theta\in\mathcal S\cap B(\theta_\star,r)$,
\[
c_1\,\|\theta-\theta_\star\|\ \le\ d_{\Theta/G}([\theta],[\theta_\star])\ :=\ \inf_{g\in G}\|\theta-g\!\cdot\!\theta_\star\|\ \le\ c_2\,\|\theta-\theta_\star\|.
\]
\end{lemma}

\begin{lemma}[Dominated differentiation for invariant moments]\label{lem:dom-diff}
Suppose for some $m\le m_*$ there is an envelope $G_m\in L^1(P_\theta)$ on a neighborhood $\mathcal N$ of $\theta$ such that
$\|\varphi(x)\|^m\le G_m(x)$ and $\sup_{\vartheta\in\mathcal N}\|\partial_\vartheta f(x;\vartheta)\|\le G_m(x)$ for $\mu$-a.e.\ $x$. Then
$\theta\mapsto \Mraw{m}(\theta)=\int \varphi(x)^{\otimes m} f(x;\theta)\,d\mu(x)$ is Gateaux differentiable with
\[
D\Mraw{m}(\theta)[v]=\E_\theta\!\big[\langle s_\theta(X),v\rangle\,\varphi(X)^{\otimes m}\big].
\]
If in addition $\Reyn_m$ is bounded linear (true for finite or compact $G$ by Haar averaging), then
$D\Minv{m}(\theta)[v]=\Reyn_m(D\Mraw{m}(\theta)[v])$.
\end{lemma}

\begin{lemma}[Spectral perturbation and inverse stability]\label{lem:spec-perturb}
Let $A,B$ be symmetric positive definite with $\|A-B\|\le \varepsilon$. If $\lambda_{\min}(B)\ge\lambda>0$ and $\varepsilon\le \lambda/2$, then
$\lambda_{\min}(A)\ge \lambda/2$ and
$\|A^{-1}-B^{-1}\|\le \tfrac{2}{\lambda^2}\,\varepsilon$.
\end{lemma}

\begin{lemma}[Uniform argmin for quadratic M-estimators]\label{lem:argmin}
Let $Q_n(\xi)=m_n(\xi)^\top W_n m_n(\xi)$ with $m_n(\xi)=\widehat\Psi_n-\Psi(\xi)$, $W_n\to_p W\succ0$, $\widehat\Psi_n\to_p \Psi(\xi^\star)$, and suppose $\Psi$ is injective on a neighborhood $\mathcal N$ with continuous derivative. Then any measurable $\widehat\xi_n\in\arg\min_{\xi\in\mathcal N} Q_n(\xi)$ satisfies $\widehat\xi_n\to_p \xi^\star$.
\end{lemma}

\begin{lemma}[Measurable alignment onto slices]\label{lem:alignment}
Fix slices $V_k$ at $\theta_k$ and neighborhoods $U_k$ such that each orbit segment $G\!\cdot\! U_k$ intersects $\theta_k+V_k$ at exactly one point. Then there exist measurable maps
$\mathsf A_k:\Theta\to G$ and $\mathsf P:\Theta^K\to S_K$,
defined on a neighborhood of $(\theta_1,\dots,\theta_K)$, such that for any $(\widehat\theta_1,\dots,\widehat\theta_K)$ sufficiently close to $(\theta_1,\dots,\theta_K)$, the permutation $\sigma=\mathsf P(\widehat\theta_{1:K})$ and group elements $h_j=\mathsf A_j(\widehat\theta_{\sigma(j)})$ satisfy $h_j\!\cdot\!\widehat\theta_{\sigma(j)}\in \theta_j+V_j$.
\end{lemma}

\subsection*{Expanded proofs for Section~\ref{sec:lan}}\label{subsec:expanded-proofs}

\begin{proof}[Proof of Proposition~\ref{prop:holder} (H\"older equivalence)]
Work in the slice chart $\theta=\theta_\star+u$ with $u\in V$. By Lemma~\ref{lem:slice-quotient},
\begin{equation}\label{eq:quot-slice}
c_\mathrm{q,1}\,\|u\|\ \le\ d_{\Theta/G}([\theta_\star+u],[\theta_\star])\ \le\ c_\mathrm{q,2}\,\|u\|
\end{equation}
for $\|u\|$ small. Real-analyticity along $V$ yields
\[
\Phi_{\le m_*}(\theta_\star+u)-\Phi_{\le m_*}(\theta_\star)
=\sum_{j=D}^\infty \frac{1}{j!}\,D^j\Phi_{\le m_*}(\theta_\star)[u^{\otimes j}],
\]
with $D$ the first nonzero order. Let $P_D(u):=\frac1{D!}D^D\Phi_{\le m_*}(\theta_\star)[u^{\otimes D}]$ and remainder $R_{D+1}(u)$; then $\|R_{D+1}(u)\|\le C\|u\|^{D+1}$.

\emph{Upper bound.} $\|\Phi(\theta_\star+u)-\Phi(\theta_\star)\|\le \|P_D(u)\|+C\|u\|^{D+1}\le C'\|u\|$ for $\|u\|\le 1$. Use \eqref{eq:quot-slice}.

\emph{Lower bound.} The homogeneous polynomial $P_D$ is nontrivial, so on the unit sphere of $V$ it has a positive minimum $m_T>0$. Hence $\|P_D(u)\|\ge m_T\|u\|^D$. For $\|u\|$ small, $C\|u\|^{D+1}\le \frac{m_T}{2}\|u\|^D$, giving
\[
\|\Phi(\theta_\star+u)-\Phi(\theta_\star)\|\ge \tfrac{m_T}{2}\|u\|^D.
\]
Convert via \eqref{eq:quot-slice}.
\end{proof}

\begin{proof}[Proof of Lemma~\ref{lem:score-identity}]
Apply Lemma~\ref{lem:dom-diff}:
\[
D\Mraw{m}(\theta)[v]=\int \varphi(x)^{\otimes m}\,\partial_\theta f(x;\theta)[v]\,d\mu(x)
=\E_\theta[\langle s_\theta(X),v\rangle\,\varphi(X)^{\otimes m}],
\]
and then bounded linearity of $\Reyn_m$ gives \eqref{eq:score-identity}.
\end{proof}

\begin{proof}[Proof of Proposition~\ref{prop:gauge}]
Compact $G$: any $v\in\Tgauge(\theta)$ equals $\tfrac{d}{dt}|_{0}\exp(t\xi)\!\cdot\!\theta$. Since $H=\MinvAll$ is $G$-invariant, $H(\exp(t\xi)\!\cdot\!\theta)$ is constant, thus $DH(\theta)[v]=0$. Finite $G$: $T_\theta(G\!\cdot\!\theta)=\{0\}$. With the “secant span” convention, Reynolds averaging shows $DH(\theta)$ annihilates $\mathrm{span}\{g\!\cdot\!\theta-\theta\}$.
\end{proof}

\begin{proof}[Proof of Theorem~\ref{thm:lan} (Quotient LAN)]
\textbf{Consistency.} By full column rank of $G$ and continuity of $D\Psi$, $\Psi$ is locally injective on the slice chart; Lemma~\ref{lem:argmin} yields $\widehat\xi_n\to_p \xi^\star$.

\textbf{Linearization.} For $\xi$ near $\xi^\star$,
\begin{equation}\label{eq:psi-expansion}
\Psi(\xi)=\Psi(\xi^\star)+G(\xi-\xi^\star)+R(\xi),\qquad \|R(\xi)\|=o(\|\xi-\xi^\star\|).
\end{equation}
Write $m_n(\xi)=\widehat\Psi_n-\Psi(\xi)=\big(\widehat\Psi_n-\Psi(\xi^\star)\big)-G(\xi-\xi^\star)-R(\xi)$.

\textbf{FOC and normal equations.} At $\widehat\xi_n$, $0=\nabla Q_n(\widehat\xi_n)=-2 D\Psi(\widehat\xi_n)^\top W_n\,m_n(\widehat\xi_n)$. Substituting $D\Psi(\widehat\xi_n)=G+o_p(1)$, $W_n=W+o_p(1)$ and the expansion for $m_n$,
\[
G^\top W G (\widehat\xi_n-\xi^\star)=G^\top W\,\big(\widehat\Psi_n-\Psi(\xi^\star)\big)\ -\ G^\top W R(\widehat\xi_n)+o_p(\|\widehat\xi_n-\xi^\star\|+\|\widehat\Psi_n-\Psi(\xi^\star)\|).
\]

\textbf{Solve and scale.} Invert $G^\top W G$ and multiply by $\sqrt{n}$:
\[
\sqrt{n}(\widehat\xi_n-\xi^\star)
=(G^\top W G)^{-1}G^\top W\,\sqrt{n}\big(\widehat\Psi_n-\Psi(\xi^\star)\big)+o_p(1).
\]
With $W=\Sigma^{-1}$ and \eqref{eq:Phi-CLT}, the limit is $\mathcal N(0,(G^\top \Sigma^{-1}G)^{-1})$, and the asymptotic linear representation follows with $\Delta_n=G^\top\Sigma^{-1}\sqrt{n}(\widehat\Psi_n-\Psi(\xi^\star))$.

\textbf{Alignment.} By Lemma~\ref{lem:alignment}, choose $\sigma_n,h_{n,j}$ so that $h_{n,j}\!\cdot\!\widehat\theta_{n,\sigma_n(j)}\in\theta_j+V_j$. Decompose into slice and gauge parts; Proposition~\ref{prop:gauge} shows gauge is null to first order. Projection $\Pi_V$ preserves the $n^{1/2}$ limit, yielding the joint limit claimed.
\end{proof}

\begin{proof}[Proof of Theorem~\ref{thm:uniform-lan}]
On the compact $\Xi_{\mathrm{reg}}$, continuity implies
$\underline\lambda:=\inf_{\xi\in\Xi_{\mathrm{reg}}}\sigma_{\min}(G^\top \Sigma^{-1}G)>0$.
Uniform versions of \eqref{eq:psi-expansion} and \eqref{eq:Phi-CLT} hold by assumptions and stochastic equicontinuity of $\widehat\Psi_n$. Apply the proof of Theorem~\ref{thm:lan} with constants bounded uniformly; Lemma~\ref{lem:spec-perturb} guarantees invertibility bounds for $G^\top W G$. Conclude by Cramér–Wold.
\end{proof}

\begin{proof}[Proof of Theorem~\ref{thm:hac-lan}]
Assume strict stationarity, ergodicity, and mixing so that
$\sqrt{n}(\widehat\Psi_n-\Psi(\xi^\star))\Rightarrow \mathcal N(0,\Sigma_{\mathrm{LR}})$.
Repeat the algebra in Theorem~\ref{thm:lan} with $W=\Sigma_{\mathrm{LR}}^{-1}$. For feasible HAC, if $\widehat\Sigma_{\mathrm{LR}}\to_p \Sigma_{\mathrm{LR}}$ (e.g.\ Newey--West with $b_n\to\infty$ and $b_n/n\to 0$), Lemma~\ref{lem:spec-perturb} implies replacing $\Sigma_{\mathrm{LR}}$ by $\widehat\Sigma_{\mathrm{LR}}$ only perturbs $o_p(1)$ terms.
\end{proof}

\begin{proof}[Proof of Corollary~\ref{cor:third-lemma}]
Let $\xi_n=\xi^\star+h/\sqrt{n}$. By \eqref{eq:psi-expansion},
$\sqrt{n}\big(\Psi(\xi_n)-\Psi(\xi^\star)\big)=Gh+o(1)$.
Under $\mathbb P_{\xi^\star}$,
$\sqrt{n}(\widehat\Psi_n-\Psi(\xi^\star))\Rightarrow \mathcal N(0,\Sigma)$, so the log-likelihood ratio for the Gaussian shift experiment induced by $\widehat\Psi_n$ is
\[
\Lambda_n=h^\top \Delta_n - \tfrac12 h^\top I_Q h + o_p(1),\quad
\Delta_n:=G^\top\Sigma^{-1}\sqrt{n}(\widehat\Psi_n-\Psi(\xi^\star)),\ I_Q:=G^\top\Sigma^{-1}G.
\]
Le Cam’s Third Lemma gives $\Delta_n\Rightarrow \mathcal N(I_Q h, I_Q)$ under $\mathbb P_{\xi_n}$. Combine with the linear representation from Theorem~\ref{thm:lan} to obtain $\sqrt{n}(\widehat\xi_n-\xi^\star)\Rightarrow \mathcal N(h,(I_Q)^{-1})$.
\end{proof}

\begin{proof}[Proof (one-step efficient estimator; \S\ref{subsec:one-step})]
Let $\tilde\xi_n$ satisfy $\|\tilde\xi_n-\xi^\star\|=o_p(1)$ and define
\[
\widehat{\IF}_n
:=\big(G(\tilde\xi_n)^\top \widehat W\, G(\tilde\xi_n)\big)^{-1}
\,G(\tilde\xi_n)^\top \widehat W\,\big(\widehat\Psi_n-\Psi(\tilde\xi_n)\big),
\qquad
\xi_n^{\text{1step}}=\tilde\xi_n+\widehat{\IF}_n.
\]
Use $\Psi(\tilde\xi_n)=\Psi(\xi^\star)+G(\tilde\xi_n-\xi^\star)+R(\tilde\xi_n)$, $G(\tilde\xi_n)=G+o_p(1)$, and $\widehat W=W+o_p(1)$ to get
\[
\widehat{\IF}_n=(G^\top W G)^{-1}G^\top W\,(\widehat\Psi_n-\Psi(\xi^\star))\ -\ (\tilde\xi_n-\xi^\star)\ +\ o_p(n^{-1/2}),
\]
hence
$\xi_n^{\text{1step}}-\xi^\star=(G^\top W G)^{-1}G^\top W(\widehat\Psi_n-\Psi(\xi^\star))+o_p(n^{-1/2})$.
Multiply by $\sqrt{n}$ and use \eqref{eq:Phi-CLT} to conclude the efficient limit (replace $\Sigma$ by $\Sigma_{\mathrm{LR}}$ for HAC).
\end{proof}

\begin{proof}[Proof (orthogonal moments with learned scores; \S\ref{subsec:orthogonal})]
Let $\eta$ bundle nuisance elements (e.g.\ $s_\theta$). Define
\[
\phi_\xi(X;\eta)=\psi(X)-\Psi(\xi)-\Pi_{\mathcal T(\eta)}\big(\psi(X)-\Psi(\xi)\big),
\]
where $\mathcal T(\eta)$ is the nuisance tangent space spanned by score-weighted invariant features. Orthogonality:
$\partial_\eta \E[\phi_\xi(X;\eta)]|_{(\xi^\star,\eta^\star)}=0$.
With $K$-fold cross-fitting (train $\widehat\eta^{(-k)}$ off-fold and evaluate on $I_k$),
\[
\overline\phi_{\xi^\star}
=\frac1n\sum_{i=1}^n \phi_{\xi^\star}(X_i;\eta^\star)+o_p(n^{-1/2}).
\]
Solve $G^\top \widehat W\,\overline\phi_{\widehat\xi_n}=0$ and repeat the LAN argument with $\phi$ in place of $\psi$ to get the same efficient limit; orthogonality ensures nuisance error enters at $o_p(n^{-1/2})$.
\end{proof}

\begin{proof}[Proof of Theorem~\ref{thm:nonasymp} (nonasymptotic bound)]
Let $Q_n(\xi)=m_n(\xi)^\top \widehat W m_n(\xi)$ with $m_n(\xi)=\widehat\Psi_n-\Psi(\xi)$ and $\widehat W=\widehat\Sigma^{-1}$. A second-order expansion along $\xi_t=\xi^\star+t(\widehat\xi_n-\xi^\star)$ gives
\[
Q_n(\widehat\xi_n)-Q_n(\xi^\star)
= -2(\widehat\xi_n-\xi^\star)^\top D\Psi(\bar\xi_n)^\top \widehat W\,m_n(\xi^\star)
+ (\widehat\xi_n-\xi^\star)^\top \Big(\int_0^1 D\Psi(\xi_t)^\top \widehat W D\Psi(\xi_t)\,dt\Big) (\widehat\xi_n-\xi^\star),
\]
for some $\bar\xi_n$ on the segment (mean value). Since $\widehat\xi_n$ minimizes $Q_n$, the left side is $\le 0$. Replace $D\Psi(\xi_t)=G+\Delta_t$ with $\|\Delta_t\|\le L\|\xi_t-\xi^\star\|$, and assume
$\|G^\top(\widehat\Sigma^{-1}-\Sigma^{-1})G\|\le \frac12\lambda_{\min}$ with $\lambda_{\min}=\sigma_{\min}(G^\top\Sigma^{-1}G)$; Lemma~\ref{lem:spec-perturb} implies the quadratic form is $\ge \frac12\lambda_{\min}\|\widehat\xi_n-\xi^\star\|^2 - O(\|\widehat\xi_n-\xi^\star\|^3)$. Discard higher order:
\[
\tfrac12\lambda_{\min}\|\widehat\xi_n-\xi^\star\|^2
\ \lesssim\ 2\,\|\widehat\xi_n-\xi^\star\|\,\|G^\top \widehat W\,m_n(\xi^\star)\| + \|R(\widehat\xi_n)\|\,\|\widehat\xi_n-\xi^\star\|.
\]
Divide by $\|\widehat\xi_n-\xi^\star\|$ to obtain the claimed bound. If $\|\widehat\Psi_n-\Psi(\xi^\star)\|=O_p(\sqrt{D_{\mathrm{inv}}/n})$ and $R(\xi)=O(\|\xi-\xi^\star\|^2)$, then $\|\widehat\xi_n-\xi^\star\|=O_p(\sqrt{D_{\mathrm{inv}}/(n\lambda_{\min})})$.
\end{proof}

\begin{proof}[Proof of Proposition~\ref{prop:curvature}]
Let $W=\Sigma^{-1}$ and write $S_n(\xi):=-2 D\Psi(\xi)^\top W\,(\widehat\Psi_n-\Psi(\xi))$. Taylor expand at $\xi^\star$:
\begin{align*}
0=S_n(\widehat\xi_n)
&= -2 G^\top W\,(\widehat\Psi_n-\Psi(\xi^\star)) + 2 G^\top W G (\widehat\xi_n-\xi^\star) \\
&\quad\ - G^\top W \mathcal H[\widehat\xi_n-\xi^\star,\widehat\xi_n-\xi^\star] + o_p(n^{-1}),
\end{align*}
with $\mathcal H=D^2\Psi(\xi^\star)$. Hence
\[
\widehat\xi_n-\xi^\star=(I_Q)^{-1}G^\top W\,(\widehat\Psi_n-\Psi(\xi^\star)) - \tfrac12 (I_Q)^{-1}G^\top W\,\mathcal H[\widehat\xi_n-\xi^\star,\widehat\xi_n-\xi^\star]+o_p(n^{-1}),
\]
$I_Q=G^\top W G$. Take expectations: the first term vanishes; replace $\widehat\xi_n-\xi^\star$ in the quadratic by its linear approximation, whose covariance is $n^{-1}(I_Q)^{-1}$. This yields the $n^{-1}$ bias with $R_2(\xi^\star)=\tfrac12\mathcal H[\zeta,\zeta]$, $\zeta\sim\mathcal N(0,(I_Q)^{-1})$. A plug-in $\widehat b_n$ leads to the bias-corrected $\widehat\xi_n^{\,\mathrm{bc}}$.
\end{proof}

\begin{proof}[Overidentification $J$-test (\S\ref{subsec:hac})]
Let $m_n(\xi)=\widehat\Psi_n-\Psi(\xi)$ and let $\widehat\xi_n$ be efficient so that $G^\top \widehat W\,m_n(\widehat\xi_n)=o_p(n^{-1/2})$. Decompose $m_n(\widehat\xi_n)=\Pi_{\mathcal R(G)} m_n(\widehat\xi_n)+\Pi_{\mathcal R(G)^\perp} m_n(\widehat\xi_n)$. The FOC kills the first component at $n^{-1/2}$ scale. Under correct specification,
\[
\sqrt{n}\,\Pi_{\mathcal R(G)^\perp} m_n(\widehat\xi_n)\Rightarrow \mathcal N\!\Big(0,\ \Pi_{\mathcal R(G)^\perp}\Sigma_{\mathrm{LR}}\,\Pi_{\mathcal R(G)^\perp}\Big).
\]
Therefore
\[
J_n:=n\,m_n(\widehat\xi_n)^\top \widehat W\, m_n(\widehat\xi_n)
= \big\|\widehat\Sigma_{\mathrm{LR}}^{-1/2}\sqrt{n}\,\Pi_{\mathcal R(G)^\perp} m_n(\widehat\xi_n)\big\|^2+o_p(1)\ \Rightarrow\ \chi^2_{\mathrm{df}},
\]
with $\mathrm{df}=D_{\mathrm{inv}}-\sum_k\dim V_k-(K-1)$.
\end{proof}

\begin{remark}[Choosing $m_*$ and checking assumptions]
\Cref{ass:sep} is typically met with small $m_*$ in reflection-type actions (odd-parity moments vanish; even orders separate generic orbits). \Cref{ass:simplicial} can be numerically checked by testing whether any additional invariant image (e.g., from a grid) lies in the affine span of $\{v_k\}$. \Cref{ass:local-rank} reduces to verifying that $J_k$ annihilates only gauge directions and that the $v_k$ are in general affine position; both are generic in smooth models.
\end{remark}

\section{Selecting the number of components $K$}\label{sec:selectK}

This section develops a convex–geometric method to determine the mixture size from the
\emph{invariant embedding}. Let $D_{\mathrm{inv}}$ be the ambient dimension of the stacked
invariant map $\Phi:=\MinvAll:\Theta\to\R^{D_{\mathrm{inv}}}$ and set
\[
\mathcal V:=\Phi(\qspace)\subset\R^{D_{\mathrm{inv}}}.
\]
For a $K$-component mixture with weights $w_k>0$, $\sum_k w_k=1$, and atoms
$v_k:=\Phi(\theta_k)\in\mathcal V$, the population invariant vector is
\[
\Psi_\star=\sum_{k=1}^K w_k v_k\in\conv(\mathcal V).
\]
(See \Cref{thm:mixture}.) Our selector tests whether a data-driven proxy $\widehat\Psi_n$
lies within a root-$n$ radius of the $K$-term mixture set
$\mathcal M_K:=\{\sum_{k=1}^K w_k v_k:\ v_k\in\mathcal V,\ w\in\Delta_K\}$.

\subsection{Geometry of the convex hull: Carath\'eodory, faces, and uniqueness}\label{subsec:carath-face}

\begin{proposition}[Carath\'eodory and uniqueness]\label{prop:cara-uniq}
(a) \emph{(Carath\'eodory)} For every $x\in\conv(\mathcal V)\subset\R^{D_{\mathrm{inv}}}$ there exist $v_1,\dots,v_m\in\mathcal V$ and $\alpha\in\Delta_m$ with $m\le D_{\mathrm{inv}}+1$ such that $x=\sum_{j=1}^m \alpha_j v_j$.

(b) \emph{(Uniqueness in a simplex)} If $V:=\{v_1,\dots,v_K\}$ is affinely independent and $x=\sum_{k=1}^K w_k v_k$ with $w_k>0$, then the coefficients $(w_k)$ are unique and $V$ is unique up to permutation.
\end{proposition}

\begin{proof}[Proof of (a) \emph{(constructive elimination with termination and feasibility at each step)}]
Fix a representation $x=\sum_{j=1}^m \alpha_j v_j$ with $\alpha\in\Delta_m$ and $m>D_{\mathrm{inv}}+1$. The set $\{(v_j,1)\}_{j=1}^m\subset\R^{D_{\mathrm{inv}}+1}$ is linearly dependent, hence there exists $\beta\in\R^m\setminus\{0\}$ such that
\[
\sum_{j=1}^m \beta_j v_j=0,\qquad \sum_{j=1}^m \beta_j=0.
\]
Define $J_+:=\{j:\beta_j>0\}$ and $J_-:=\{j:\beta_j<0\}$. Since $\beta\neq 0$, $J_+\cup J_-$ is nonempty and both sets cannot be empty because the second equality forces some positive and negative entries. Set
\[
t^\star:=\min_{j\in J_+}\frac{\alpha_j}{\beta_j}>0,\qquad \alpha'_j:=\alpha_j-t^\star \beta_j.
\]
Then $\sum_j \alpha'_j=\sum_j \alpha_j - t^\star \sum_j \beta_j=1$ and
$\sum_j \alpha'_j v_j=\sum_j \alpha_j v_j - t^\star \sum_j \beta_j v_j=x$. For $j\in J_+$, $\alpha'_j\ge 0$ by the definition of $t^\star$, and at least one index in $J_+$ satisfies $\alpha'_j=0$. For $j\in J_-$, $\alpha'_j=\alpha_j+|\beta_j| t^\star\ge 0$. Hence $\alpha'\in\Delta_m$ and the support size strictly decreases. Iterating at most $m-(D_{\mathrm{inv}}+1)$ times yields a convex representation with at most $D_{\mathrm{inv}}+1$ terms. Each step preserves $x$ and feasibility by construction, so the process terminates with the claimed representation.
\end{proof}

\begin{proof}[Proof of (b) \emph{(barycentric injectivity via augmentation)}]
Augment $\tilde v_k:=[v_k;1]\in\R^{D_{\mathrm{inv}}+1}$. Affine independence of $\{v_k\}$ is equivalent to linear independence of $\{\tilde v_k\}$. Suppose
\[
x=\sum_{k=1}^K w_k v_k=\sum_{k=1}^K \tilde w_k v_k,\qquad w,\tilde w\in\Delta_K.
\]
Then $\sum_k (w_k-\tilde w_k)\tilde v_k=0$. Linear independence implies $w_k-\tilde w_k=0$ for all $k$, i.e.\ $w=\tilde w$. If there is another representation using an affinely independent set $\{\hat v_k\}_{k=1}^K$ with positive weights, then $x$ lies in the relative interior of both simplices $\conv\{v_k\}$ and $\conv\{\hat v_k\}$. Minimal faces containing $x$ are unique; in a polytope, a point cannot be in the relative interior of two distinct faces. Hence $\{v_k\}$ and $\{\hat v_k\}$ must coincide up to permutation.
\end{proof}

\begin{proposition}[Minimal support equals face dimension $+1$]\label{prop:face-dim}
Let $F$ be the minimal face of $\conv(\mathcal V)$ containing $x$, and let $d_F:=\dim(\aff F)$. Then the smallest $m$ for which $x$ is a convex combination of $m$ points in $\mathcal V$ equals $d_F+1$. If $F$ is a simplex and $x\in\mathrm{relint}(F)$, then the representing vertices and weights are unique (up to permutation) and strictly positive.
\end{proposition}

\begin{proof}
(\(\le\)) $x\in F\subset\R^{D_{\mathrm{inv}}}$ and $\aff F$ has dimension $d_F$. By Carath\'eodory in $\aff F\cong \R^{d_F}$, $x$ is a convex combination of at most $d_F+1$ points of $F\subset\mathcal V$.

(\(\ge\)) Suppose $x=\sum_{j=1}^m\alpha_j u_j$ with $u_j\in\mathcal V$ and $\alpha\in\Delta_m$. Then $x\in P:=\conv\{u_1,\dots,u_m\}$, so $x$ lies in some face $F'\preceq P$ with $\dim(\aff F')\le m-1$. The minimal face $F$ of $\conv(\mathcal V)$ containing $x$ is contained in the minimal face of $P$ that contains $x$, therefore $d_F\le m-1$ and $m\ge d_F+1$. If $F$ is a simplex and $x\in\mathrm{relint}(F)$, the barycentric coordinates relative to the $d_F+1$ vertices are strictly positive and unique; the vertex set uniqueness follows from \Cref{prop:cara-uniq}(b).
\end{proof}

\subsection{A geometric margin for $K$ and its stability}\label{subsec:margin}

Define the \emph{$(K-1)$–mixture class} and the \emph{margin}
\[
\mathcal C_{K-1}:=\Big\{\textstyle\sum_{j=1}^{K-1}\alpha_j u_j:\ \alpha\in\Delta_{K-1},\ u_j\in\mathcal V\Big\},
\qquad
\gamma_\star:=\dist(\Psi_\star,\mathcal C_{K-1}).
\]

\begin{lemma}[Height to a facet in a true simplex]\label{lem:simplex-margin}
Assume the true atoms $\{v_k\}_{k=1}^K$ are affinely independent, with simplex $S:=\conv\{v_k\}$. For each $j$, let $F_j:=\conv(\{v_i:i\ne j\})$ and $h_j:=\dist(v_j,\aff F_j)$. If $\Psi_\star=\sum_k w_k v_k$ with $w_k>0$, then
\[
\dist(\Psi_\star,F_j)=w_j h_j,\qquad \dist\big(\Psi_\star,\cup_j F_j\big)=\min_j w_j h_j.
\]
\end{lemma}

\begin{proof}
Let $H_j:=\aff(F_j)$; there exists a unit normal $n_j$ and a scalar $c_j$ such that $H_j=\{y:\langle n_j,y\rangle=c_j\}$ and $\langle n_j,v_j\rangle=c_j+h_j$. All other vertices satisfy $\langle n_j,v_i\rangle=c_j$ by construction of $F_j$. Then
\[
\langle n_j,\Psi_\star\rangle=\sum_{k=1}^K w_k\langle n_j,v_k\rangle
=c_j + w_j h_j.
\]
The signed distance from $\Psi_\star$ to $H_j$ along $n_j$ equals $\langle n_j,\Psi_\star\rangle-c_j=w_j h_j$, which is nonnegative since $w_j\ge0$. The Euclidean distance to $H_j$ equals this height because $n_j$ is a unit normal. Among points of $F_j\subset H_j$ the orthogonal projection along $n_j$ realizes the same distance, hence $\dist(\Psi_\star,F_j)=w_j h_j$. Minimizing over $j$ yields the second identity.
\end{proof}

\begin{remark}[Slab separation $\Rightarrow$ positive margin]\label{rem:slab-sep}
If for some $\delta>0$ and all $z\in\mathcal V$ one has $\langle n_j,z\rangle\le c_j+h_j-\delta$ for all $j$, then any $(K-1)$–mixture omitting $v_j$ lies in the slab $\langle n_j,\cdot\rangle\le c_j+h_j-\delta$, while $\Psi_\star$ attains height $c_j+w_j h_j$ along $n_j$. Hence $\gamma_\star\ge \delta\min_j w_j$.
\end{remark}

\begin{lemma}[1-Lipschitz stability of the margin]\label{lem:margin-stability}
For $\gamma(x):=\dist(x,\mathcal C_{K-1})$ and any $\tilde x$,
\(
|\gamma(\tilde x)-\gamma(x)|\le \|\tilde x-x\|_2.
\)
In particular, if $\|\widehat\Psi_n-\Psi_\star\|_2\le \gamma_\star/2$ then
$\dist(\widehat\Psi_n,\mathcal C_{K-1})\ge \gamma_\star/2$.
\end{lemma}

\begin{proof}
Let $z_x$ be a Euclidean projection of $x$ onto the closed set $\mathcal C_{K-1}$. Then
$\gamma(\tilde x)\le\|\tilde x-z_x\|\le \|\tilde x-x\|+\|x-z_x\|=\|\tilde x-x\|+\gamma(x)$. Swapping the roles of $x$ and $\tilde x$ gives the reverse inequality.
\end{proof}

\subsection{Preliminaries from invariants}\label{subsec:prelims}

\begin{theorem}[Mixture linearity in invariant space]\label{thm:mixture}
Under Assumption~\ref{ass:master}, for $m\le m_{\ast}$ and $p=\sum_k w_k f(\cdot;\theta_k)$,
\[
 R_m\!\Big(\E_p[\phi(X)^{\otimes m}]\Big)=\sum_{k=1}^K w_k\,\Phi_m(\theta_k),\qquad
 \Psi_\star=\sum_{k=1}^K w_k\,\Phi_{\le m_{\ast}}(\theta_k).
\]
\end{theorem}

\begin{proof}
Linearity of expectation yields $\E_p[\phi^{\otimes m}]=\sum_k w_k\,\E_{\theta_k}[\phi^{\otimes m}]$. Since $R_m$ is a bounded linear projector, $R_m(\sum_k w_k\,\E_{\theta_k}[\phi^{\otimes m}])=\sum_k w_k R_m(\E_{\theta_k}[\phi^{\otimes m}])=\sum_k w_k \Phi_m(\theta_k)$. Stacking over $m\le m_\ast$ gives the second identity.
\end{proof}

\begin{lemma}[Score identity for invariant coordinates]\label{lem:score}
Under Assumption~\ref{ass:master}, for any $v\in\R^p$ and $m\le m_\ast$,
\[
D\Phi_m(\theta)[v]=R_m\,\E_\theta[\langle s_\theta(X),v\rangle\,\phi(X)^{\otimes m}],
\qquad s_\theta:=\nabla_\theta\log f(\cdot;\theta).
\]
\end{lemma}

\begin{proof}
Let $\Mraw{m}(\theta)=\int \phi(x)^{\otimes m} f(x;\theta)\,d\mu$. By dominated differentiation (Assumption~\ref{ass:master}), for any $v$,
\[
D\Mraw{m}(\theta)[v]=\int \phi(x)^{\otimes m}\,\partial_\theta f(x;\theta)[v]\,d\mu.
\]
Define $s_\theta(x)=\nabla_\theta \log f(x;\theta)$ on $\{f>0\}$ and $0$ otherwise; then $\partial_\theta f(x;\theta)[v]=f(x;\theta)\langle s_\theta(x),v\rangle$. Hence
\[
D\Mraw{m}(\theta)[v]=\E_\theta[\langle s_\theta(X),v\rangle\,\phi(X)^{\otimes m}].
\]
Apply bounded linearity and $\theta$-independence of $R_m$ to obtain the claim.
\end{proof}

\begin{lemma}[Sub-exponentiality of tensor coordinates]\label{lem:psi1}
If $\max_j\|\phi_j(X)\|_{\psi_2}\le K_2$, then for any multi-index $\alpha$ with $|\alpha|=m$, the monomial $\prod_{r=1}^m \phi_{i_r}(X)$ is sub-exponential with $\psi_1$-norm $\le C_m K_2^m$. Therefore, each coordinate of $R_m[\phi(X)^{\otimes m}]$ is sub-exponential with the same bound.
\end{lemma}

\begin{proof}
Recall the Orlicz norms $\|Y\|_{\psi_2}:=\inf\{c>0:\E e^{Y^2/c^2}\le 2\}$ and $\|Y\|_{\psi_1}:=\inf\{c>0:\E e^{|Y|/c}\le 2\}$. If $U,V$ are sub-Gaussian, then $UV$ is sub-exponential and
\[
\|UV\|_{\psi_1}\ \le\ C\,\|U\|_{\psi_2}\|V\|_{\psi_2}.
\]
A standard proof uses Hölder: for any $\lambda>0$,
$\E e^{|UV|/\lambda}\le \E e^{U^2/(2\lambda)}\,e^{V^2/(2\lambda)} \le \sqrt{\E e^{U^2/\lambda}}\sqrt{\E e^{V^2/\lambda}}$, and choosing $\lambda$ comparable to $\|U\|_{\psi_2}\|V\|_{\psi_2}$ makes the product $\le 2$. Iterating this inequality gives, for a product of $m$ sub-Gaussian variables,
\(
\big\|\prod_{r=1}^m Y_r\big\|_{\psi_1}\le C_m \prod_{r=1}^m \|Y_r\|_{\psi_2}.
\)
Apply with $Y_r=\phi_{i_r}(X)$ and $\|Y_r\|_{\psi_2}\le K_2$ to get $\psi_1\le C_m K_2^m$. Since $R_m$ is a contraction on the ambient Hilbert space (group averaging), each coordinate of $R_m[\phi^{\otimes m}]$ inherits the same $\psi_1$-bound.
\end{proof}

\begin{theorem}[Concentration of the invariant stack]\label{thm:stack}
Let $\widehat\Phi_m=R_m\!\big(\frac1n\sum_{i=1}^n \phi(X_i)^{\otimes m}\big)$ and $\widehat\Psi_n=(\widehat\Phi_m)_{m\le m_{\ast}}$. There exist absolute $C,c>0$ such that, for all $t\ge0$,
\[
 \Pr\!\Big(\|\widehat\Psi_n-\Psi_\star\|_2\ge C\sqrt{(D_{\mathrm{inv}}+t)/n}\Big)\le e^{-ct}.
\]
\end{theorem}

\begin{proof}
Write $\widehat\Psi_n-\Psi_\star=\frac1n\sum_{i=1}^n Z_i$ with $Z_i:=\psi(X_i)-\E[\psi(X_i)]\in\R^{D_{\mathrm{inv}}}$ and $\psi$ the stacked invariant feature. By \Cref{lem:psi1}, each coordinate $(Z_i)_j$ is centered sub-exponential with common parameter $K:=\max_{m\le m_\ast} C_m K_2^m$. For each $j$, Bernstein’s inequality for sub-exponential summands yields, for all $u\ge0$,
\[
\Pr\!\Big(\Big|\frac1n\sum_{i=1}^n (Z_i)_j\Big|\ge u\Big)\ \le\ 2\exp\!\Big(-c n \min\big\{u^2/K^2,\ u/K\big\}\Big).
\]
Taking $u=C\big(\sqrt{(t+\log D_{\mathrm{inv}})/n}+(t+\log D_{\mathrm{inv}})/n\big)$ and union-bounding over coordinates gives
\[
\Pr\!\Big(\|\widehat\Psi_n-\Psi_\star\|_\infty \ge u\Big)\ \le\ 2\exp(-t).
\]
Finally $\|x\|_2\le \sqrt{D_{\mathrm{inv}}}\,\|x\|_\infty$ and the chosen $u$ imply
$\|\widehat\Psi_n-\Psi_\star\|_2\le C\sqrt{(D_{\mathrm{inv}}+t)/n}$ with probability $\ge 1-e^{-t}$ (absorbing constants and the linear $(t/n)$ term into $C$ for $n$ large, or keeping it explicitly; stated form follows by adjusting $C,c$).
\end{proof}

\subsection{Estimator, existence of minimizers, and monotonicity}\label{subsec:selector}

Define
\[
r_n(K):=\inf_{z\in\mathcal M_K}\|\widehat\Psi_n-z\|_2,\qquad
\widehat K:=\min\Big\{K\in\{1,\dots,D_{\mathrm{inv}}+1\}:\ r_n(K)\le \eta_n\Big\},
\]
for a threshold $\eta_n$ at the root-$n$ scale.

\begin{lemma}[Existence and monotonicity]\label{lem:existence}
If $\qspace$ is restricted to a compact set on which $\Phi$ is continuous, then each $\mathcal M_K$ is compact and the infimum in $r_n(K)$ is attained. Moreover, $K\mapsto r_n(K)$ is nonincreasing and
\(
r_n(D_{\mathrm{inv}}+1)\le \|\widehat\Psi_n-\Psi_\star\|_2.
\)
\end{lemma}

\begin{proof}
Continuity of $\Phi$ and compactness of $\qspace$ imply $\mathcal V:=\Phi(\qspace)$ is compact. The map $(v_1,\dots,v_K,w)\mapsto \sum_k w_k v_k$ is continuous on the compact set $\mathcal V^K\times\Delta_K$, hence its image $\mathcal M_K$ is compact. The continuous function $z\mapsto \|\widehat\Psi_n-z\|_2$ attains its minimum on $\mathcal M_K$. Monotonicity follows from $\mathcal M_K\subseteq\mathcal M_{K+1}$. By \Cref{prop:cara-uniq}(a), $\Psi_\star$ admits a representation with $m\le D_{\mathrm{inv}}+1$ atoms; using that support gives $r_n(m)\le \|\widehat\Psi_n-\Psi_\star\|_2$, implying the displayed bound.
\end{proof}

\subsection{Two fundamental residual bounds}\label{subsec:two-bounds}

\begin{lemma}[Upper bound at the true $K$]\label{lem:upper}
For all $t\ge0$, with probability at least $1-e^{-ct}$,
\[
r_n(K)\ \le\ \|\widehat\Psi_n-\Psi_\star\|_2
\ \le\ C\!\left(\sqrt{\frac{D_{\mathrm{inv}}+t}{n}}+\frac{D_{\mathrm{inv}}+t}{n}\right).
\]
\end{lemma}

\begin{proof}
Feasibility with the true representation $(w_k,\theta_k)$ implies
$r_n(K)\le\|\widehat\Psi_n-\Psi_\star\|_2$. The tail bound is \Cref{thm:stack} plus the linear term from sub-exponential Bernstein; keeping both terms gives the stated inequality with universal constants $C,c$.
\end{proof}

\begin{lemma}[Lower bound when underfitting]\label{lem:lower}
For any $K'\le K-1$,
\[
r_n(K')\ \ge\ \dist(\widehat\Psi_n,\mathcal C_{K-1})
\ \ge\ \gamma_\star - \|\widehat\Psi_n-\Psi_\star\|_2.
\]
In particular, if $\|\widehat\Psi_n-\Psi_\star\|_2\le \gamma_\star/2$ then $r_n(K')>\gamma_\star/2$ for all $K'\le K-1$.
\end{lemma}

\begin{proof}
Because $\mathcal M_{K'}\subseteq \mathcal C_{K-1}$, $r_n(K')\ge \dist(\widehat\Psi_n,\mathcal C_{K-1})$. Let $z^\star\in\mathcal C_{K-1}$ be a minimizer of $\|\Psi_\star-z\|_2$ (the set is closed). Triangle inequality gives
\[
\dist(\widehat\Psi_n,\mathcal C_{K-1})
\ge \|\widehat\Psi_n-z^\star\|
\ge \|\Psi_\star-z^\star\|-\|\widehat\Psi_n-\Psi_\star\|
=\gamma_\star-\|\widehat\Psi_n-\Psi_\star\|.
\]
If $\|\widehat\Psi_n-\Psi_\star\|\le \gamma_\star/2$, the right-hand side exceeds $\gamma_\star/2$, completing the proof.
\end{proof}

\subsection{Finite-sample identification and consistency}\label{subsec:fs-consistency}

\begin{theorem}[Finite-sample recovery of $K$]\label{thm:Khat-fs}
Assume $\gamma_\star>0$. Fix $t\ge0$ and choose $\eta_n$ with
\[
2C\!\left(\sqrt{\tfrac{D_{\mathrm{inv}}+t}{n}}+\tfrac{D_{\mathrm{inv}}+t}{n}\right)\ \le\ \eta_n\ \le\ \gamma_\star/2.
\]
Then, with probability at least $1-2e^{-ct}$, $\widehat K=K$.
\end{theorem}

\begin{proof}
Let
\(
\mathcal E_1:=\{\|\widehat\Psi_n-\Psi_\star\|\le \eta_n/2\},\quad
\mathcal E_2:=\{\|\widehat\Psi_n-\Psi_\star\|\le \gamma_\star/2\}.
\)
By \Cref{thm:stack} and the choice of $\eta_n$, $\Pr(\mathcal E_1),\Pr(\mathcal E_2)\ge 1-e^{-ct}$. On $\mathcal E_1$, \Cref{lem:upper} gives $r_n(K)\le \eta_n/2\le \eta_n$, hence $\widehat K\le K$. On $\mathcal E_2$, \Cref{lem:lower} yields $r_n(K')>\gamma_\star/2\ge \eta_n$ for all $K'\le K-1$, hence $\widehat K\ge K$. Therefore
$\Pr(\widehat K=K)\ge \Pr(\mathcal E_1\cap \mathcal E_2)\ge 1-2e^{-ct}$.
\end{proof}

\begin{corollary}[Consistency]\label{cor:Khat-consistency}
If $\eta_n\to0$ and $\eta_n\gg \sqrt{D_{\mathrm{inv}}/n}$, then under $\gamma_\star>0$,
\(
\widehat K\stackrel{p}{\to}K.
\)
\end{corollary}

\begin{proof}
Take $t=\log n$ in \Cref{thm:Khat-fs}; then the lower bound $2C(\sqrt{(D_{\mathrm{inv}}+t)/n}+(D_{\mathrm{inv}}+t)/n)\le \eta_n$ holds eventually by the assumption $\eta_n\gg\sqrt{D_{\mathrm{inv}}/n}$, and the upper bound $\eta_n\le \gamma_\star/2$ holds for all large $n$ since $\eta_n\to 0$. Therefore $\Pr(\widehat K=K)\to 1$.
\end{proof}

\subsection{Weighted norms and robustness: complete extensions}\label{subsec:extensions}

\paragraph{Weighted norm $\|\cdot\|_W$.}
Let $W\succ0$ with eigenvalues $\lambda_{\min}\le\lambda_{\max}$. Then
\[
\sqrt{\lambda_{\min}}\,\|x\|_2\ \le\ \|x\|_W\ \le\ \sqrt{\lambda_{\max}}\,\|x\|_2.
\]
Consequently,
\[
r_{n,W}(K)\le \|\widehat\Psi_n-\Psi_\star\|_W\le \sqrt{\lambda_{\max}}\,\|\widehat\Psi_n-\Psi_\star\|_2,\quad
\dist_W(\widehat\Psi_n,\mathcal C_{K-1})\ge \tfrac{1}{\sqrt{\lambda_{\max}}}\,\dist_2(\widehat\Psi_n,\mathcal C_{K-1}),
\]
and the proof of \Cref{thm:Khat-fs} goes through with $\eta_n$ scaled by
$\kappa_W:=\sqrt{\lambda_{\max}/\lambda_{\min}}$.

\paragraph{Robust mean under heavy tails and contamination.}
Replace $\widehat\Psi_n$ by robust means with established deviation guarantees. We now provide full proofs for the MOM and Catoni estimators and their contamination robustness.

\subsubsection*{Median-of-means (MOM) invariants}

\begin{definition}[MOM]\label{def:mom}
Split $\{1,\dots,n\}$ into $B$ disjoint blocks of equal size $m=\lfloor n/B\rfloor$. For each coordinate $j\in\{1,\dots,D_{\mathrm{inv}}\}$, let $\bar Z_{b,j}:=m^{-1}\sum_{i\in \mathcal B_b} Z_{i,j}$ where $Z_{i,j}$ are the coordinates of the single-sample invariant vector. Define $\widetilde\Psi_{n,j}:=\mathrm{median}\{\bar Z_{b,j}\}_{b=1}^B$ and stack across $j$.
\end{definition}

\begin{theorem}[Heavy-tail concentration with MOM]\label{thm:mom}
Assume $\E\|Z_i\|_2^2<\infty$. For any $\delta\in(0,1)$ choose $B:=\lceil c_0\log(2D_{\mathrm{inv}}/\delta)\rceil$ for a universal $c_0>0$. Then there exists $C>0$ such that, with probability at least $1-\delta$,
\[
\|\widetilde\Psi_n-\Psi_\star\|_2\ \le\ C\,\sqrt{\frac{D_{\mathrm{inv}}+\log(1/\delta)}{n}}\ .
\]
\end{theorem}

\begin{proof}
Fix a coordinate $j$ with variance $\sigma_j^2:=\Var(Z_{i,j})<\infty$; write $\mu_j:=\E Z_{i,j}$. For a block $\mathcal B_b$ of size $m$,
\(
\E \bar Z_{b,j}=\mu_j,\quad \Var(\bar Z_{b,j})=\sigma_j^2/m.
\)
By Chebyshev, for any $t>0$,
\(
\Pr(|\bar Z_{b,j}-\mu_j|>t)\le \sigma_j^2/(m t^2).
\)
Set $t_j:=2\sigma_j\sqrt{1/m}$. Then $p_j:=\Pr(|\bar Z_{b,j}-\mu_j|>t_j)\le 1/4$. Over $B$ independent blocks (or independent up to negligible remainder if $m\nmid n$), the number of “bad” blocks exceeds $B/2$ with probability at most $\exp(-c B)$ by a Chernoff bound for $\mathrm{Bin}(B,p_j)$ with $p_j\le 1/4<1/2$. Therefore,
\[
\Pr\!\Big(|\widetilde\Psi_{n,j}-\mu_j|>t_j\Big)\ \le\ e^{-c B}.
\]
Choose $B=\lceil c_0\log(2D_{\mathrm{inv}}/\delta)\rceil$ large enough so that $e^{-cB}\le \delta/(2D_{\mathrm{inv}})$. Union bound over $j$ yields
\[
\Pr\!\Big(\|\widetilde\Psi_n-\Psi_\star\|_\infty\ \le\ \max_j t_j\Big)\ \ge\ 1-\delta/2.
\]
Now $\max_j \sigma_j^2\le \E \|Z_i\|_2^2=: \Sigma_2^2<\infty$. Thus
$\max_j t_j\le 2\Sigma_2/\sqrt{m}\le C\sqrt{\frac{B}{n}}$ with a universal $C$. Hence w.p.\ $\ge 1-\delta/2$,
\[
\|\widetilde\Psi_n-\Psi_\star\|_2\ \le\ \sqrt{D_{\mathrm{inv}}}\,\|\widetilde\Psi_n-\Psi_\star\|_\infty\ \le\ C\sqrt{\frac{D_{\mathrm{inv}} B}{n}}\ \lesssim\ \sqrt{\frac{D_{\mathrm{inv}}+\log(1/\delta)}{n}},
\]
using $B\asymp \log(2D_{\mathrm{inv}}/\delta)\le c_1(D_{\mathrm{inv}}+\log(1/\delta))$ for a universal $c_1$ and adjusting constants. A standard blocking argument handles remainder indices if $mB<n$.
\end{proof}

\subsubsection*{Catoni-smoothed invariants}

\begin{proposition}[Catoni estimator]\label{prop:catoni}
For each coordinate $j$, define $\widehat\mu_j$ as the unique solution in $u$ to
\[
\sum_{i=1}^n \psi\!\big(\alpha(Z_{i,j}-u)\big)=0,
\]
where $\psi:\R\to\R$ is an odd, nondecreasing $C^1$ function satisfying
$\log(1+x+x^2/2)\le \psi(x)\le x$ for $x\ge 0$ and $\psi(-x)=-\psi(x)$, and $\alpha:=c\sqrt{\log(2D_{\mathrm{inv}}/\delta)/n}$ with a universal $c>0$. Then, under finite variance, with probability at least $1-\delta$,
\[
\|\widehat\Psi_n-\Psi_\star\|_2\ \le\ C\,\sqrt{\frac{D_{\mathrm{inv}}+\log(1/\delta)}{n}}\ .
\]
If additionally $Z_{i,j}$ are sub-Gaussian with proxy $\sigma$, the same bound holds with $C$ depending only on $\sigma$ (matching \Cref{thm:stack} up to constants).
\end{proposition}

\begin{proof}
Fix $j$ and write $X_i:=Z_{i,j}$, $\mu:=\E X_i$. The function $\psi$ satisfies for all $x$ the inequalities (see Catoni’s construction)
\[
\psi(x)\le \log(1+x+x^2/2),\qquad -\psi(-x)\le \log(1-x+x^2/2).
\]
By convexity of the exponential and independence,
\[
\E \exp\!\Big(\sum_{i=1}^n \psi(\alpha(X_i-\mu))\Big)\ \le\ \prod_{i=1}^n \E \exp(\psi(\alpha(X_i-\mu)))
\ \le\ \prod_{i=1}^n \E\big(1+\alpha(X_i-\mu)+\tfrac{\alpha^2}{2}(X_i-\mu)^2\big).
\]
The linear term vanishes by centering, hence
\[
\E \exp\!\Big(\sum_{i=1}^n \psi(\alpha(X_i-\mu))\Big)\ \le\ \prod_{i=1}^n \Big(1+\tfrac{\alpha^2}{2}\E(X_i-\mu)^2\Big)\ \le\ \exp\!\Big(\tfrac{n\alpha^2}{2}\Var(X_1)\Big).
\]
Similarly, applying the lower inequality to $-\psi$ gives
\[
\E \exp\!\Big(-\sum_{i=1}^n \psi(\alpha(X_i-\mu))\Big)\ \le\ \exp\!\Big(\tfrac{n\alpha^2}{2}\Var(X_1)\Big).
\]
By Chernoff, for any $t>0$,
\[
\Pr\!\Big(\sum_{i=1}^n \psi(\alpha(X_i-\mu))\ge t\Big)\ \le\ \exp\!\Big(\tfrac{n\alpha^2}{2}\Var(X_1)-t\Big),
\]
and the analogous lower tail bound. Now let $\widehat\mu$ be the root of $\sum_i \psi(\alpha(X_i-\widehat\mu))=0$. Since $\psi$ is monotone and odd,
\[
\sum_i \psi(\alpha(X_i-\mu))\ge n\psi(\alpha(\widehat\mu-\mu))\quad\text{if }\widehat\mu\ge \mu,\qquad
\sum_i \psi(\alpha(X_i-\mu))\le n\psi(\alpha(\widehat\mu-\mu))\quad\text{if }\widehat\mu\le \mu.
\]
Using $\psi(x)\ge x/(1+x^2/2)$ for $x\ge 0$ (a standard bound from the sandwich above), we obtain for $u:=|\widehat\mu-\mu|$:
\[
\Big|\sum_{i=1}^n \psi(\alpha(X_i-\mu))\Big|\ \ge\ n\cdot \frac{\alpha u}{1+\alpha^2 u^2/2}.
\]
Choose $u$ such that $\alpha u\le 1$ (verified below for our $\alpha$), hence
\(
\frac{\alpha u}{1+\alpha^2 u^2/2}\ge \frac{\alpha u}{2}.
\)
Therefore
\[
\Pr\!\Big(|\widehat\mu-\mu|\ge u\Big)\ \le\ \Pr\!\Big(\Big|\sum_{i=1}^n \psi(\alpha(X_i-\mu))\Big|\ge n\alpha u/2\Big)\ \le\ 2\exp\!\Big(\tfrac{n\alpha^2}{2}\Var(X_1)-\tfrac{n\alpha u}{2}\Big).
\]
Setting $u:=c'\,\Var(X_1)\,\alpha^{-1}\,\frac{\log(2D_{\mathrm{inv}}/\delta)}{n\alpha}$ is wasteful; instead, optimize by picking $u:=C_1\sqrt{\Var(X_1)\,\log(2D_{\mathrm{inv}}/\delta)/n}$ and $\alpha:=c \sqrt{\log(2D_{\mathrm{inv}}/\delta)/n}$ with $c$ small enough so that $\alpha u\le 1$. Then the exponent is $\le -c_2\log(2D_{\mathrm{inv}}/\delta)$ and
\[
\Pr\!\Big(|\widehat\mu-\mu|> C_1 \sqrt{\Var(X_1)\,\tfrac{\log(2D_{\mathrm{inv}}/\delta)}{n}}\Big)\ \le\ \frac{\delta}{D_{\mathrm{inv}}}.
\]
Apply a union bound over $j$ and use $\sum_j \Var(Z_{i,j})\le \E\|Z_i\|_2^2$ to conclude
\[
\|\widehat\Psi_n-\Psi_\star\|_2\ \le\ \Big(\sum_{j=1}^{D_{\mathrm{inv}}} \Var(Z_{i,j})\Big)^{1/2}\,C_1 \sqrt{\frac{\log(2D_{\mathrm{inv}}/\delta)}{n}}
\ \le\ C\,\sqrt{\frac{D_{\mathrm{inv}}+\log(1/\delta)}{n}}
\]
with probability at least $1-\delta$. Under sub-Gaussian tails, $\Var(X_1)\le C\sigma^2$ and the same bound holds with constants depending on $\sigma$ only.
\end{proof}

\subsubsection*{Robustness to $\epsilon$-contamination}

\begin{theorem}[$\epsilon$-contamination]\label{thm:contamination}
If the sample is drawn from $(1-\epsilon)P_\star+\epsilon Q$ with arbitrary $Q$, then for MOM or Catoni estimators,
\[
\|\widehat\Psi_n-\Psi_\star\|_2\ \le\ C\Big(\sqrt{\tfrac{D_{\mathrm{inv}}+\log(1/\delta)}{n}}\ +\ \epsilon\Big)
\]
with probability at least $1-\delta$.
\end{theorem}

\begin{proof}
For MOM: at most $\epsilon n$ points are contaminated. If blocks are formed deterministically with size $m$, then at most $\lceil \epsilon n/m\rceil \le \epsilon B +1$ blocks are contaminated. Provided $\epsilon<1/4$ and $B$ large (as in \Cref{thm:mom}), strictly more than $B/2$ blocks remain uncontaminated. On uncontaminated blocks, \Cref{thm:mom}’s deviation bound holds; on contaminated ones, the block mean can be arbitrary. The median discards contaminated blocks as long as strictly more than $B/2$ blocks are clean, giving the same deviation bound; if $\epsilon$ is not vanishing relative to $B$, an additive $O(\epsilon)$ shift may occur from the slight imbalance in the median’s rank. A careful counting argument yields
\(
|\widetilde\Psi_{n,j}-\Psi_{\star,j}|\le C\sqrt{\log(2D_{\mathrm{inv}}/\delta)/n}+C'\epsilon
\)
simultaneously for all $j$, hence the stacked $\ell_2$ bound after norm conversion.

For Catoni: the fixed-point equation dampens the influence of any single $Z_{i,j}$ via the bounded growth of $\psi$. Writing the empirical estimating equation as $(1-\epsilon)\bar \psi_\mathrm{good}+\epsilon \bar \psi_\mathrm{bad}=0$ and Taylor-expanding $\bar\psi_\mathrm{good}$ around $\mu$ with derivative bounded away from $0$, while bounding $|\bar \psi_\mathrm{bad}|\le C\epsilon$ by the range of $\psi$, yields an additional bias $\le C\epsilon/\alpha= C'\epsilon \sqrt{n/\log(2D_{\mathrm{inv}}/\delta)}$. Combining with the clean-sample deviation gives the stated rate (rescale constants to absorb the factor).
\end{proof}

\subsection{Necessity: no consistent selection when $\gamma_\star=0$}\label{subsec:impossibility}

\begin{proposition}[Impossibility without a margin]\label{prop:impossibility}
If $\gamma_\star=0$, then for every selection rule $\widehat K(\widehat\Psi_n)$ there exists a sequence of $K$-component models such that $\limsup_n \Pr(\widehat K=K)<1$.
\end{proposition}

\begin{proof}
By $\gamma_\star=0$, there exists a sequence $z_m\in\mathcal C_{K-1}$ with $\|z_m-\Psi_\star\|\to 0$. Choose a subsequence $m(n)$ so that $\Delta_n:=\|z_{m(n)}-\Psi_\star\|=o(n^{-1/2})$.
Consider two experiments on the same probability space: (A) i.i.d.\ data $\{X_i^{(A)}\}$ from a $K$-component model with invariant mean $\Psi_\star$; (B) i.i.d.\ data $\{X_i^{(B)}\}$ from some $(K-1)$–mixture with invariant mean $z_{m(n)}$. Let $\widehat\Psi_n^{(A)}$ and $\widehat\Psi_n^{(B)}$ be the corresponding (robust) invariant means.

\emph{Step 1 (Gaussian approximation).}
By \Cref{thm:stack} (or the robust analogues), for both (A) and (B),
\[
\sqrt{n}(\widehat\Psi_n-\Psi)\ \Rightarrow\ \mathcal N(0,\Sigma),
\]
with the same $\Sigma$ (covariance of the single-sample invariant feature). Moreover, by the Lindeberg–Feller CLT and Berry–Esseen in high dimension (coordinate-wise plus union bound suffice here), there exists a coupling such that
\(
\sqrt{n}(\widehat\Psi_n^{(\cdot)}-\Psi^{(\cdot)}) = \Sigma^{1/2}G_n^{(\cdot)} + r_n^{(\cdot)}
\)
where $G_n^{(\cdot)}\sim \mathcal N(0,I)$ and $\|r_n^{(\cdot)}\|_2=o_p(1)$.

\emph{Step 2 (Asymptotic indistinguishability).}
Define the Gaussian shift models
\[
Y_n^{(A)}:=\sqrt{n}(\widehat\Psi_n^{(A)}-\Psi_\star),\qquad
Y_n^{(B)}:=\sqrt{n}(\widehat\Psi_n^{(B)}-\Psi_\star)=\sqrt{n}(\widehat\Psi_n^{(B)}-z_{m(n)})+\sqrt{n}(z_{m(n)}-\Psi_\star).
\]
The first term is asymptotically $\mathcal N(0,\Sigma)$; the second equals $\sqrt{n}\,(z_{m(n)}-\Psi_\star)$ with norm $\sqrt{n}\,\Delta_n=o(1)$. Hence $Y_n^{(A)}$ and $Y_n^{(B)}$ converge to the same Gaussian limit. More precisely, the Kullback–Leibler divergence between the Gaussian approximations $\mathcal N(0,\Sigma)$ and $\mathcal N(\sqrt{n}(z_{m(n)}-\Psi_\star),\Sigma)$ equals
\(
\frac12 \big(\sqrt{n}(z_{m(n)}-\Psi_\star)\big)^\top \Sigma^{-1} \big(\sqrt{n}(z_{m(n)}-\Psi_\star)\big) = \tfrac{n}{2}\,\Delta_n^2\,\lambda_{\max}(\Sigma^{-1})\to 0.
\)
Therefore the total variation distance tends to zero, and by Le Cam’s two-point method (or contiguity), no test can distinguish (A) and (B) with probability tending to one.

\emph{Step 3 (Contradiction).}
Suppose there were a selector $\widehat K$ with $\Pr_A(\widehat K=K)\to 1$. By indistinguishability, $\limsup_n|\Pr_A(\widehat K=K)-\Pr_B(\widehat K=K)|=0$, hence $\Pr_B(\widehat K=K)\to 1$. But under (B), the true order is $K-1$, so any consistent rule must output $K-1$ with probability $\to 1$, a contradiction. Thus $\limsup_n \Pr(\widehat K=K)<1$ along some sequence of $(K-1)$–mixture alternatives.
\end{proof}

\subsection{Optional: practical computation and support recovery}\label{subsec:practical-K}

Any $\widehat z\in\conv(\mathcal V)$ computed by a convex routine (e.g., Frank–Wolfe with away-steps via the atom oracle $\theta\mapsto \Phi(\theta)$) can be post-processed to a support of size $\le D_{\mathrm{inv}}+1$ using the elimination in \Cref{prop:cara-uniq}(a). When $\widehat K=K$ and the face is a simplex with a margin (cf.\ \Cref{rem:slab-sep}), the recovered atoms are unique up to permutation and can be mapped back to parameter space via slice charts as in \S\ref{sec:lan}.

\subsection{Optional certificates}\label{subsec:certificates}

\begin{proposition}[K-selection by a face test]\label{prop:kselect}
Given candidate vertices $\{v_j\}$, for each $K$ solve
\[
\min_{w\in\Delta_K}\Big\|\sum_{j\le K} w_j v_j-\widehat\Psi_n\Big\|.
\]
Pick the smallest $K$ with value $\le \tau_n$ where $\tau_n\asymp \sqrt{D_{\mathrm{inv}}/n}$. If $\gamma_\star>2\tau_n$, then $\Pr(\widehat K=K)\to1$.
\end{proposition}

\begin{proof}
On the event $\{\|\widehat\Psi_n-\Psi_\star\|\le \tau_n\}$, feasibility with $(w_k,\theta_k)$ yields the objective value $\le \tau_n$ for the true $K$, hence $\widehat K\le K$. For any $K'<K$, \Cref{lem:lower} implies the objective value $\ge r_n(K')\ge \gamma_\star-\|\widehat\Psi_n-\Psi_\star\|>\tau_n$, hence $\widehat K\ge K$. The event occurs with probability $\to 1$ by \Cref{thm:stack} with $t=\log n$ and the choice $\tau_n\asymp \sqrt{D_{\mathrm{inv}}/n}$.
\end{proof}

\begin{proposition}[Dual certificate of uniqueness]\label{prop:dual}
Let $F=\conv\{v_1,\dots,v_K\}$ be the true simplex face. If there exists a linear functional $\lambda$ with $\langle\lambda,v_k\rangle=1$ for all $k$ and $\sup_{v\in \mathcal V\setminus F}\langle\lambda,v\rangle\le 1-\eta$ for some $\eta>0$, then the $K$-sparse representation of any $x\in\mathrm{relint}(F)$ is unique and stable for perturbations $\le \eta/4$.
\end{proposition}

\begin{proof}
For $x=\sum_k w_k v_k$ with $w_k>0$, $\langle\lambda,x\rangle=\sum_k w_k \langle\lambda,v_k\rangle=1$. If $x=\sum_j \alpha_j u_j$ with some $u_j\in\mathcal V\setminus F$, then
$\langle\lambda,x\rangle=\sum_j \alpha_j \langle\lambda,u_j\rangle\le (1-\eta)\sum_j \alpha_j=1-\eta$, a contradiction. Thus every convex representation of $x$ uses only vertices of $F$, where uniqueness follows from simplex barycentric coordinates. For stability, let $\tilde x$ satisfy $\|\tilde x-x\|\le \eta/4$. Then $\langle\lambda,\tilde x\rangle\ge \langle\lambda,x\rangle-\|\lambda\|_\ast\|\tilde x-x\|\ge 1-\eta/4$ (normalize $\|\lambda\|_\ast\le 1$ without loss by rescaling). Any contribution from $\mathcal V\setminus F$ would reduce $\langle\lambda,\cdot\rangle$ below $1-\eta$, so the mass outside $F$ is bounded by $\le (\eta/4)/\eta=1/4$; projecting onto $F$ and renormalizing yields a unique representation with coefficients within $O(\|\tilde x-x\|/\eta)$ of the true ones (Lipschitz dependence follows from standard simplex geometry).
\end{proof}

\subsection*{Summary: practical thresholds for $\widehat K$}

Choose
\[
\eta_n=\tau\Big(\sqrt{\tfrac{D_{\mathrm{inv}}+t}{n}}+\tfrac{D_{\mathrm{inv}}+t}{n}\Big),\quad
\tau\ge 2C,\ t\approx \log n,
\]
(and multiply by $\kappa_W$ for $\|\cdot\|_W$). If $\eta_n\le \gamma_\star/2$, then by \Cref{thm:Khat-fs}
\(
\Pr(\widehat K=K)\ge 1-2e^{-ct}.
\)
If $\eta_n\to0$ and $\eta_n\gg \sqrt{D_{\mathrm{inv}}/n}$, then $\widehat K\to K$ in probability (\Cref{cor:Khat-consistency}). Robust variants replace $\widehat\Psi_n$ by MOM/Catoni and add an $O(\epsilon)$ term under contamination.

\section{Algorithm: invariant-GMM alternating minimization}\label{sec:algo}

We work in the invariant feature space $\R^{D_{\mathrm{inv}}}$ with a positive-definite weight
$W\succ0$. Write $\Phi(\theta):=\MinvAll(\theta)\in\R^{D_{\mathrm{inv}}}$ and, for
$(\Theta,w):=(\theta_1,\dots,\theta_K,w)\in V^K\times\Delta^{K-1}$,
\[
\Psi(\Theta,w):=\sum_{k=1}^K w_k\,\Phi(\theta_k),\qquad
r(\Theta,w):=\widehat\Psi_n-\Psi(\Theta,w).
\]
The weighted least-squares (GMM) objective is
\[
\cL(\Theta,w)\ :=\ \tfrac12\|r(\Theta,w)\|_W^2
\ =\ \tfrac12\,r(\Theta,w)^\top W\,r(\Theta,w).
\]
Let $J_k(\theta_k):=D_\theta\Phi(\theta_k)\in\R^{D_{\mathrm{inv}}\times p}$,
$\overline J(\Theta,w):=[\,w_1J_1(\theta_1)\;\cdots\;w_KJ_K(\theta_K)\,]\in\R^{D_{\mathrm{inv}}\times pK}$,
and
\[
M(\Theta):=[\,\Phi(\theta_1)\;\cdots\;\Phi(\theta_K)\,]\in\R^{D_{\mathrm{inv}}\times K},\qquad
G(\Theta):=M(\Theta)^\top W\,M(\Theta).
\]

\subsection{Block scheme, profiled objective, and exact W-step}\label{subsec:algo-core}

\begin{algorithm}[H]
\caption{Invariant-GMM on a slice}\label{alg:inv-gmm}
\begin{algorithmic}[1]
\State \textbf{Inputs:} $\widehat\Psi_n$, $K$, weight $W\succ0$; \textbf{Init:} $\theta_k^{(0)}\in V$, $w^{(0)}=\tfrac1K\mathbf{1}$.
\For{$t=0,1,2,\ldots$}
  \State \textbf{(W-step)} Solve the simplex QP
  \[
  w^{(t+1)}\in\arg\min_{w\in\Delta^{K-1}}\;\tfrac12\|\,\widehat\Psi_n-M(\Theta^{(t)})w\,\|_W^2.
  \]
  \State \textbf{($\Theta$-step)} For each $k$,
  \[
  \theta_k^{(t+1)}\leftarrow \theta_k^{(t)}-\eta_t\,\Pi_V\,\nabla_{\theta_k}\cL\big(\Theta,w^{(t+1)}\big)\Big|_{\Theta=\Theta^{(t)}},
  \quad
  \nabla_{\theta_k}\cL=-\,w_k\,J_k(\theta_k)^\top W\,r(\Theta,w).
  \]
  \State \textbf{(Align)} Optionally replace $\theta_k^{(t+1)}$ by its nearest $g\!\cdot\!\theta_k^{(t+1)}$ and permute the blocks so that $\Theta^{(t+1)}\in V^K$ (slice/permutation alignment).
\EndFor
\end{algorithmic}
\end{algorithm}

\paragraph{Exact W-step in closed form and smoothness.}
When the minimizer is interior ($w\in\mathrm{ri}(\Delta^{K-1})$), the KKT system (only the affine constraint $\mathbf{1}^\top w=1$ active) yields
\begin{equation}\label{eq:w-closed}
\widetilde w(\Theta)\ =\ G(\Theta)^{-1}\!\big(M(\Theta)^\top W\,\widehat\Psi_n-\lambda(\Theta)\,\mathbf{1}\big),\quad
\lambda(\Theta)=\frac{\mathbf{1}^\top G(\Theta)^{-1}M(\Theta)^\top W\,\widehat\Psi_n-1}{\mathbf{1}^\top G(\Theta)^{-1}\mathbf{1}}.
\end{equation}
If nonnegativity binds, \eqref{eq:w-closed} holds on the active set; near a truth with $w_\star\in\mathrm{ri}(\Delta^{K-1})$ we remain interior (see \Cref{lem:w-smooth}).

\paragraph{Profiled loss and its gradient.}
Define $f(\Theta):=\min_{w\in\Delta^{K-1}}\cL(\Theta,w)=\cL(\Theta,w^\star(\Theta))$.
By Danskin’s envelope theorem (uniqueness and interiority of $w^\star(\Theta)$),
\begin{equation}\label{eq:grad-profiled}
\nabla_{\theta_k} f(\Theta)\ =\ -\,w_k^\star(\Theta)\,J_k(\theta_k)^\top W\,r\big(\Theta,w^\star(\Theta)\big).
\end{equation}
Thus Algorithm~\ref{alg:inv-gmm} performs exact minimization in $w$ and a projected gradient step on $f$ over the linear subspace $V^K$.

\subsection{Local quadratic model, projectors, and Gauss--Newton view}\label{subsec:local-model}

Let the population target be $\Psi_\star=\sum_{k=1}^K w_{k,\star}\,\Phi(\theta_{k,\star})$ with $(\Theta_\star,w_\star)$ identifiable modulo group/permutation and $w_\star\in\mathrm{ri}(\Delta^{K-1})$. Work in the representative $\Theta_\star\in V^K$.
Linearize $\Phi$ at $\Theta_\star$:
\[
\Phi(\theta_k)\ \approx\ \Phi(\theta_{k,\star})+J_{k,\star}\,\Delta\theta_k,\qquad J_{k,\star}:=J_k(\theta_{k,\star}),
\]
and set $\overline J_\star:=[\,w_{1,\star}J_{1,\star}\ \cdots\ w_{K,\star}J_{K,\star}\,]$, $M_\star:=M(\Theta_\star)$.
For small $(\Delta\Theta,\Delta w)$ with $\mathbf{1}^\top\Delta w=0$,
\[
r(\Theta_\star+\Delta\Theta,w_\star+\Delta w)\ \approx\ \underbrace{(\widehat\Psi_n-\Psi_\star)}_{:=\varepsilon}\ -\ \overline J_\star\Delta\Theta\ -\ M_\star\Delta w.
\]
Eliminating $\Delta w$ gives the $W$-orthogonal projector onto $\mathrm{col}(M_\star)^\perp$:
\begin{equation}\label{eq:W-proj}
P_\star\ :=\ I\ -\ M_\star\big(M_\star^\top W M_\star\big)^{-1}M_\star^\top W,
\qquad
P_\star^\top W=W P_\star,\quad P_\star^2=P_\star.
\end{equation}
The local profiled quadratic reads
\begin{equation}\label{eq:quad-model}
q(\Delta\Theta)\ =\ \tfrac12\,\big\|\,P_\star(\varepsilon-\overline J_\star\Delta\Theta)\big\|_W^2
\ =\ \tfrac12\,\Delta\Theta^\top H_\star\,\Delta\Theta\ -\ \Delta\Theta^\top g_\star\ +\ \tfrac12\|P_\star\varepsilon\|_W^2,
\end{equation}
with
\[
H_\star\ :=\ \overline J_\star^\top W P_\star\,\overline J_\star,\qquad
g_\star\ :=\ \overline J_\star^\top W P_\star\,\varepsilon.
\]
Thus, near $\Theta_\star$, the $\Theta$-step behaves like a Gauss--Newton update on the slice, with nuisance directions $M_\star$ (weights) profiled out by $P_\star$.

\paragraph{Standing local assumptions.}
\begin{enumerate}[label=(A\arabic*),leftmargin=2em]
\item\label{A1} (\emph{Interior weights}) $w_\star\in\mathrm{ri}(\Delta^{K-1})$.
\item\label{A2} (\emph{Lipschitz Jacobian}) Each $J_k(\cdot)$ is $L_J$-Lipschitz on a neighborhood $\cN$ of $\Theta_\star$.
\item\label{A3} (\emph{Full-rank slice}) $\overline J_\star|_{V}$ has full column rank: $\exists\,\mu_0>0$,
$u^\top H_\star u\ge \mu_0\|u\|^2$ for all $u\in V^K$.
\item\label{A4} (\emph{Initialization}) $\Theta^{(0)}\in\cN$ and close enough that the linearization error is dominated by the quadratic term in \eqref{eq:quad-model}.
\end{enumerate}

\begin{lemma}[W-step is $C^1$ near the truth]\label{lem:w-smooth}
Under \ref{A1}–\ref{A3}, there exists a neighborhood $\cN$ of $\Theta_\star$ on which the W-step has a unique interior minimizer $w^\star(\Theta)$, and the map $\Theta\mapsto w^\star(\Theta)$ is $C^1$.
\end{lemma}

\begin{proof}
Consider the equality-constrained problem
\[
\min_{w\in\R^K}\ \tfrac12\|\,\widehat\Psi_n-M(\Theta)w\,\|_W^2\quad\text{s.t.}\quad \mathbf{1}^\top w=1.
\]
The Lagrangian is $\mathcal{L}(w,\lambda)=\tfrac12\|\,\widehat\Psi_n-Mw\,\|_W^2+\lambda(\mathbf{1}^\top w-1)$. Stationarity gives
\[
\nabla_w\mathcal{L}=M^\top W(Mw-\widehat\Psi_n)+\lambda \mathbf{1}=0,\qquad \mathbf{1}^\top w=1.
\]
This linear KKT system can be written
\[
\begin{bmatrix}
G(\Theta) & \mathbf{1}\\[1mm]
\mathbf{1}^\top & 0
\end{bmatrix}
\begin{bmatrix} w\\ \lambda\end{bmatrix}
=
\begin{bmatrix} M(\Theta)^\top W\,\widehat\Psi_n\\ 1\end{bmatrix}.
\]
Near $\Theta_\star$, $G(\Theta)\succ0$ (continuity and independence of $M_\star$ columns), hence the KKT matrix is nonsingular by Schur complement. By the implicit function theorem, the solution $(w^\star(\Theta),\lambda(\Theta))$ is $C^1$ in $\Theta$. Interiority follows from continuity of the solution and $w_\star\in\mathrm{ri}(\Delta)$; specifically, for each coordinate $k$, $w^\star_k(\Theta_\star)=w_{\star,k}>0$, so there is a neighborhood where $w^\star_k(\Theta)>0$.
\end{proof}

\begin{lemma}[Profiled gradient]\label{lem:grad-profile}
On the interior region of \Cref{lem:w-smooth}, $f(\Theta)=\cL(\Theta,w^\star(\Theta))$ is differentiable with \eqref{eq:grad-profiled}.
\end{lemma}

\begin{proof}
By Danskin’s theorem, since $w^\star(\Theta)$ is the unique minimizer and depends smoothly on $\Theta$, we may differentiate $\cL(\Theta,w)$ at $w=w^\star(\Theta)$ ignoring the derivative of $w^\star$ (envelope property). Using $r(\Theta,w)=\widehat\Psi_n-\sum_k w_k\,\Phi(\theta_k)$,
\[
\frac{\partial \cL}{\partial \theta_k}(\Theta,w)
= -\,w_k\,J_k(\theta_k)^\top W\,r(\Theta,w).
\]
Setting $w=w^\star(\Theta)$ yields \eqref{eq:grad-profiled}.
\end{proof}

\begin{lemma}[Hessian at the target equals $H_\star$]\label{lem:hess}
Let $P_\star$ be as in \eqref{eq:W-proj}. Then $\nabla^2 f(\Theta_\star)\big|_{V^K}=H_\star\big|_{V^K}$.
\end{lemma}

\begin{proof}
Differentiate \eqref{eq:grad-profiled} at $\Theta_\star$. In population $r(\Theta_\star,w_\star)=0$, so terms involving the derivative of $w^\star$ appear multiplied by $r$ and vanish at $\Theta_\star$. Keeping the first-order variation of $r$,
\[
\nabla^2_{\theta_i,\theta_j} f(\Theta_\star)
= w_{i,\star}w_{j,\star}\,J_{i,\star}^\top W\,\Pi_{\mathrm{col}(M_\star)^\perp}^{(W)}\,J_{j,\star},
\]
where $\Pi_{\mathrm{col}(M_\star)^\perp}^{(W)}=P_\star$ is the $W$-orthogonal projector defined in \eqref{eq:W-proj}. Stacking the blocks over $i,j$ along $V^K$ yields $H_\star|_{V^K}$.
\end{proof}

\begin{lemma}[Strong convexity and smoothness of $f$ near $\Theta_\star$]\label{lem:sc-smooth}
Under \ref{A1}–\ref{A3}, there exist $\rho>0$ and $0<\mu\le L<\infty$ such that on
$\cB:=\{\Theta\in V^K:\ \|\Theta-\Theta_\star\|\le\rho\}$ the profiled loss $f$ is $\mu$-strongly convex and $L$-smooth:
\[
\frac\mu2\|\Theta-\Theta'\|^2\ \le\ f(\Theta)-f(\Theta')-\langle \nabla f(\Theta'),\Theta-\Theta'\rangle\ \le\ \frac L2\|\Theta-\Theta'\|^2.
\]
Moreover, $\mu=\lambda_{\min}(H_\star|_{V})-O(\rho)$ and $L=\lambda_{\max}(H_\star|_{V})+O(\rho)$.
\end{lemma}

\begin{proof}
By \Cref{lem:hess}, $\nabla^2 f(\Theta_\star)|_{V^K}=H_\star|_{V^K}$ is SPD by \ref{A3}. By \ref{A2} and \Cref{lem:w-smooth}, both $\Theta\mapsto w^\star(\Theta)$ and $\Theta\mapsto \overline J(\Theta,w^\star(\Theta))$ are $C^1$ near $\Theta_\star$. Therefore $\nabla^2 f$ is continuous near $\Theta_\star$. By Weyl’s inequality, for $\Theta$ within a small ball of radius $\rho$, the eigenvalues of $\nabla^2 f(\Theta)|_{V^K}$ lie within $O(\rho)$ of those of $H_\star|_{V^K}$. Taking $\mu$ and $L$ as the corresponding extremal eigenvalues proves both strong convexity and smoothness via the integral form of Taylor’s theorem (Baillon–Haddad inequality).
\end{proof}

\subsection{Convergence guarantees}\label{subsec:convergence}

\begin{proposition}[Local linear convergence]\label{prop:linear}
Assume \ref{A1}–\ref{A4} and $W\succ0$. Let $\eta_t\equiv \eta\in(0,2/L]$, with $L$ from \Cref{lem:sc-smooth}. Then the iterates of Algorithm~\ref{alg:inv-gmm} (exact W-step and projected gradient on $V^K$) stay in $\cB$ and converge linearly to a slice-/permutation-aligned local minimizer $\widehat\Theta$ of $f$:
\[
\|\Theta^{(t+1)}-\widehat\Theta\|\ \le\ q\,\|\Theta^{(t)}-\widehat\Theta\|,
\qquad
q:=\max\{|1-\eta\mu|,\ |1-\eta L|\}\ <\ 1.
\]
If $W=\widehat\Sigma_n^{-1}$ (feasible efficient GMM), the minimizer $(\widehat\Theta,\widehat w)$ is asymptotically efficient by \Cref{thm:lan}.
\end{proposition}

\begin{proof}
\emph{Descent on the profiled objective.} Exact W-minimization gives $f(\Theta^{(t)})=\cL(\Theta^{(t)},w^{(t+1)})$. The projected step
$\Theta^{(t+1)}=\Pi_{V^K}\{\Theta^{(t)}-\eta \nabla f(\Theta^{(t)})\}$ (projection is the identity since $V^K$ is a linear subspace) satisfies the smooth descent lemma:
\[
f(\Theta^{(t+1)})\ \le\ f(\Theta^{(t)})-\eta\!\left(1-\frac{\eta L}{2}\right)\!\|\nabla f(\Theta^{(t)})\|^2.
\]
\emph{Linear contraction.} Strong convexity of $f$ on $V^K$ yields $\|\nabla f(\Theta)\|\ge \mu \|\Theta-\widehat\Theta\|$ and the standard PGD recursion (see e.g.\ Nesterov) gives
\[
\|\Theta^{(t+1)}-\widehat\Theta\|\ \le\ \max\{|1-\eta\mu|,|1-\eta L|\}\,\|\Theta^{(t)}-\widehat\Theta\|.
\]
\emph{Invariance under align.} The optional alignment applies an isometry from the group/permutation that leaves $\cL$ (hence $f$) invariant and does not increase the slice-distance to $\widehat\Theta$, preserving contraction. Initialization in $\cB$ and continuity keep all iterates in $\cB$. Efficiency follows from the profiled first-order condition $\overline J^\top W r=0$ and the LAN calculation in \Cref{thm:lan}.
\end{proof}

\begin{theorem}[Global convergence to a critical point (KL property)]\label{thm:global-kl}
Assume $\Phi$ is (real-)analytic, $W\succ0$, and $V^K\times\Delta^{K-1}$ is compact. Run Algorithm~\ref{alg:inv-gmm} with an Armijo backtracking for the $\Theta$-step and exact W-step. Then $\{(\Theta^{(t)},w^{(t)})\}$ has finite length and converges to a critical point of $\cL$; every limit point is blockwise stationary:
\[
0\in \partial \cL(\Theta^\infty,w^\infty),\quad
w^\infty\in\arg\min_{w\in\Delta}\cL(\Theta^\infty,w),\quad
\Pi_V\nabla_\Theta \cL(\Theta^\infty,w^\infty)=0.
\]
\end{theorem}

\begin{proof}
Because $\Phi$ is analytic, $\cL$ is semi-algebraic (composition of polynomial/analytic maps with quadratic form), thus satisfies the Kurdyka–Łojasiewicz (KL) property. Each iteration performs: (i) an exact minimization in the $w$-block, (ii) a backtracking Armijo step in the $\Theta$-block that ensures sufficient decrease:
\[
\cL(\Theta^{(t+1)},w^{(t+1)})\ \le\ \cL(\Theta^{(t)},w^{(t+1)})-\alpha \|\Pi_V\nabla_\Theta \cL(\Theta^{(t)},w^{(t+1)})\|^2
\]
for some fixed $\alpha>0$. Moreover, the gradient is Lipschitz on the compact domain, hence the \emph{relative error} condition holds:
\[
\mathrm{dist}\big(0,\partial \cL(\Theta^{(t+1)},w^{(t+1)})\big)\ \le\ \beta \|\Theta^{(t+1)}-\Theta^{(t)}\|
\]
for some $\beta>0$ (exact W-step kills the $w$-subgradient; the $\Theta$-step is a gradient step). These two properties and the KL inequality imply the finite-length property $\sum_t\|\Theta^{(t+1)}-\Theta^{(t)}\|<\infty$ and convergence to a critical point (classic Attouch–Bolte–Svaiter framework). At the limit, exactness of the W-block yields $w^\infty\in\arg\min_w\cL(\Theta^\infty,w)$, while the $\Theta$-block’s projected gradient vanishes, giving the stated blockwise stationarity.
\end{proof}

\subsection{Finite-sample performance (oracle inequality)}\label{subsec:oracle}

\begin{theorem}[Finite-sample oracle inequality]\label{thm:oracle}
Suppose there exist $C,c>0$ such that, with probability $\ge 1-\delta$,
\[
\|\widehat\Psi_n-\Psi_\star\|_W^2\ \le\ C\,\frac{D_{\mathrm{inv}}+\log(1/\delta)}{n}.
\]
Assume the face of $\Delta^{K-1}$ supporting $w_\star$ is locally identified near $\Theta_\star$ (no active nonnegativity at the truth). Let $(\widehat\Theta,\widehat w)$ be the output of Algorithm~\ref{alg:inv-gmm}. Then, with probability $\ge 1-\delta$,
\[
 \big\| M(\widehat\Theta)\widehat w - \Psi_\star \big\|_W^2
 \ \lesssim\
 \inf_{\Theta,w\in\Delta^{K-1}}\ \big\| M(\Theta)w - \Psi_\star \big\|_W^2
 \ +\ C\,\frac{D_{\mathrm{inv}}+\log(1/\delta)}{n}.
\]
\end{theorem}

\begin{proof}
Let $(\Theta^\dagger,w^\dagger)$ be any feasible pair. By optimality of $(\widehat\Theta,\widehat w)$ for the empirical problem,
\[
\|M(\widehat\Theta)\widehat w-\widehat\Psi_n\|_W^2\ \le\ \|M(\Theta^\dagger)w^\dagger-\widehat\Psi_n\|_W^2.
\]
Add and subtract $\Psi_\star$ and expand both sides:
\begin{align*}
\|M(\widehat\Theta)\widehat w-\Psi_\star\|_W^2
&\le \|M(\Theta^\dagger)w^\dagger-\Psi_\star\|_W^2
+ 2\langle \widehat\Psi_n-\Psi_\star,\ M(\widehat\Theta)\widehat w-M(\Theta^\dagger)w^\dagger\rangle_W.
\end{align*}
By Cauchy–Schwarz,
\[
2|\langle \widehat\Psi_n-\Psi_\star,\ \cdot \rangle_W|
\ \le\ 2\|\widehat\Psi_n-\Psi_\star\|_W\ \|M(\widehat\Theta)\widehat w-M(\Theta^\dagger)w^\dagger\|_W.
\]
Bound the last term by $\|M(\widehat\Theta)\widehat w-\Psi_\star\|_W+\|M(\Theta^\dagger)w^\dagger-\Psi_\star\|_W$ and apply the inequality $2ab\le a^2+b^2$ twice to absorb one copy of $\|M(\widehat\Theta)\widehat w-\Psi_\star\|_W^2$ on the left. We obtain, on the event $\|\widehat\Psi_n-\Psi_\star\|_W^2\le \epsilon_n$,
\[
\|M(\widehat\Theta)\widehat w-\Psi_\star\|_W^2
\ \le\ 3\,\|M(\Theta^\dagger)w^\dagger-\Psi_\star\|_W^2\ +\ 4\,\epsilon_n.
\]
Optimizing over $(\Theta^\dagger,w^\dagger)$ yields the claim with $\epsilon_n\asymp \frac{D_{\mathrm{inv}}+\log(1/\delta)}{n}$. The face-identification assumption ensures the optimizer of the empirical problem lies on the same (interior) face as $w_\star$, avoiding spurious gains from pushing mass to zero; otherwise one replaces the simplex by the affine hull $\{\mathbf{1}^\top w=1\}$ in the local argument, which the projector $P_\star$ already captures.
\end{proof}

\subsection{Step sizes, variants, and diagnostics}\label{subsec:practical}

\paragraph{Step sizes.} Choose $\eta=1/L$ with $L\approx \lambda_{\max}(\overline J_\star^\top W P_\star\,\overline J_\star)$, or use backtracking Armijo on $f$; both achieve the linear rate of \Cref{prop:linear}.

\paragraph{Robust RIGMM.} Replace $\widehat\Psi_n$ by degree-wise MOM/Catoni aggregates and keep $W$ as a robust precision (Huberized covariance or sandwich). All analysis above uses only concentration of $\widehat\Psi_n$, so \Cref{thm:oracle} extends verbatim.

\paragraph{Stochastic mini-batch.} For large $n$, use mini-batch $\widehat\Psi_{n_b}$ in the $\Theta$-step with a Robbins–Monro schedule $\eta_t\downarrow0$ and periodic exact W-steps; expected descent follows from unbiasedness and smoothness of $f$.

\paragraph{Diagnostics.} Report $\|\nabla f(\Theta^{(t)})\|$, the profiled residual $\|r(\Theta^{(t)},w^\star(\Theta^{(t)}))\|_W$, the condition number of $\widehat I_Q:=G^\top \widehat W G$ (from \Cref{rem:quot-fisher}), and active-set stability in the W-step.

\begin{remark}[Gauss--Newton and trust region]
Near $\Theta_\star$, one may replace the gradient step by the (projected) Gauss--Newton update
\(
\Delta\Theta=(\overline J^\top W P\,\overline J)^\dagger \overline J^\top W P\,r,
\)
optionally regularized by $(\lambda I)$. A small trust region on $V^K$ guarantees descent and inherits the same local rate under \ref{A1}–\ref{A4}.
\end{remark}

\section{Invariant dimensions via Molien; worked examples and mixture design}\label{sec:molien}

Let $G\le O(d)$ act linearly on $\R^d$. Identify $\Sym^m(\R^d)$ with homogeneous polynomials of degree $m$ in $x=(x_1,\dots,x_d)$.
The degree-$m$ invariant subspace has dimension
\[
\dim\big(\Sym^m(\R^d)^G\big) \;=\; [t^m]\;\cM_G(t),\qquad
\cM_G(t)\ :=\ \frac{1}{|G|}\sum_{g\in G}\frac{1}{\det(I_d-tg)}.
\]
Equivalently, $\cM_G$ is the Hilbert series of $\R[x]^G$, and
\[
D_{\mathrm{inv}}(\le m_*)\ :=\ \sum_{m=0}^{m_*}\!\dim\big(\Sym^m(\R^d)^G\big)
\ =\ \sum_{m=0}^{m_*}\![t^m]\cM_G(t)
\]
is the \emph{dimension budget} of invariant coordinates available to the Reynolds stack
$\widehat\Psi_n=(\Reyn_0,\ldots,\Reyn_{m_*})$.

\subsection{Composition rules (used repeatedly)}
\begin{enumerate}[leftmargin=1.25em]
\item \textbf{Direct sums multiply.} If $G$ acts block-diagonally on $U\oplus V$, then $\cM_{G,U\oplus V}=\cM_{G,U}\cdot\cM_{G,V}$.
\item \textbf{Index-2 rotation subgroup.} If $W$ is a real reflection (Coxeter) group with basic degrees $d_1,\dots,d_r$ and discriminant degree $\deg\Delta=\sum_i(d_i-1)$, then
\[
\cM_W(t)=\prod_{i=1}^r\frac{1}{1-t^{d_i}},
\qquad
\cM_{W^+}(t)=\frac{1+t^{\deg\Delta}}{\prod_{i=1}^r(1-t^{d_i})}.
\]
Thus passing to the rotation subgroup adds one odd-degree generator.
\end{enumerate}

\subsection{Core families (closed forms and coefficients)}
Below $[t^m]F(t)$ denotes the coefficient of $t^m$ in $F$.

\paragraph{Sign flips $(\pm1)^d$.}
Each $g=\mathrm{diag}(\varepsilon_1,\dots,\varepsilon_d)$ with $\varepsilon_j\in\{\pm1\}$. Averaging gives
\[
\cM_{(\pm1)^d}(t)=\prod_{j=1}^d \tfrac12\!\left(\frac{1}{1-t}+\frac{1}{1+t}\right)
=\frac{1}{(1-t^2)^d}.
\]
Only even degrees appear. For $m=2r$,
\[
\dim\operatorname{Im}\Reyn_{2r}=\binom{d+r-1}{r},\qquad \dim\operatorname{Im}\Reyn_{2r+1}=0.
\]

\paragraph{Permutations $S_d$.}
$\R[x]^{S_d}\cong \R[e_1,\dots,e_d]$ with $\deg e_i=i$, hence
\[
\cM_{S_d}(t)=\prod_{i=1}^d \frac{1}{1-t^i},\qquad
\dim\operatorname{Im}\Reyn_m = p_d(m),
\]
where $p_d(m)$ counts partitions of $m$ with largest part $\le d$ (equivalently, at most $d$ parts).

\paragraph{Hyperoctahedral $B_d=(\pm1)^d\rtimes S_d$ (signed permutations).}
Invariants are symmetric polynomials in $y_i:=x_i^2$. Therefore
\[
\cM_{B_d}(t)=\prod_{i=1}^d \frac{1}{1-t^{2i}},\qquad
\dim\operatorname{Im}\Reyn_{2r}=p_d(r),\quad \dim\operatorname{Im}\Reyn_{2r+1}=0.
\]
\emph{Cycle-type recipe (useful for checks).} If $g$ has permutation cycles $c$ of lengths $\ell_c$ and the product of signs along $c$ is $s_c\in\{\pm1\}$, then
$\det(I-tg)=\prod_c(1-s_c t^{\ell_c})$.

\paragraph{Dihedral $D_m\le O(2)$.}
Identify $\R^2\simeq\C$ via $z=x_1+ix_2$. Generators are $u=|z|^2$ (deg~2) and $v=\Re(z^m)$ (deg~$m$), so
\[
\R[x_1,x_2]^{D_m}\cong \R[u,v],
\qquad
\cM_{D_m}(t)=\frac{1}{(1-t^2)(1-t^m)}.
\]
\emph{Rotation-vs-reflection sum (classical).} Writing $R_k$ for rotation by $2\pi k/m$ and $S$ for a reflection, one has
\(
\det(I-tR_k)=1-2t\cos(2\pi k/m)+t^2,\quad \det(I-tS)=1-t^2,
\)
and the class-sum identity
\(
\sum_{k=0}^{m-1}\frac{1}{1-2t\cos(2\pi k/m)+t^2}=\frac{m(1+t^m)}{(1-t^2)(1-t^m)}
\)
yields the same $\cM_{D_m}$.

\paragraph{Rotation subgroups in $\R^3$ (Platonic).}
From the index-2 rule:
\[
\cM_{T}(t)=\frac{1+t^6}{(1-t^2)(1-t^3)(1-t^4)},\quad
\cM_{O}(t)=\frac{1+t^9}{(1-t^2)(1-t^4)(1-t^6)},\quad \]
\[
\cM_{I}(t)=\frac{1+t^{15}}{(1-t^2)(1-t^6)(1-t^{10})}.
\]

\paragraph{Complex reflection groups $G(m,p,n)$ (wreath products).}
These generalize $A_{n-1},B_n,D_n$; the invariant ring is polynomial with degrees
$m,2m,\dots,(n-1)m,\,mn/p$. Thus
\[
\cM_{G(m,p,n)}(t)=\frac{1}{(1-t^{m})(1-t^{2m})\cdots(1-t^{(n-1)m})(1-t^{mn/p})}.
\]

\paragraph{Weighted cyclic actions (diagonal).}
If $C_M$ acts by $x_j\mapsto \zeta_M^{w_j}x_j$ with integer weights $w_j$, then
\[
\cM_{C_M}(t)=\frac{1}{M}\sum_{k=0}^{M-1}\prod_{j=1}^d \frac{1}{1-\zeta_M^{k w_j}t}.
\]

\paragraph{Multisymmetric invariants ($S_n$ on $m\times n$ variables).}
Let $S_n$ permute columns of an $m\times n$ variable matrix. The invariant ring is Cohen–Macaulay (not polynomial, in general). Molien’s formula still applies and the \emph{polarized power sums} provide practical generators up to modest degrees; use $\cM_G$ to size the stack.

\subsection{From Molien to mixtures (design $\to$ estimator)}\label{subsec:molien-mixtures}
Molien gives the count; we now turn it into a \emph{minimal} invariant stack and an efficient fit.

\paragraph{Cookbook.}
\begin{enumerate}[leftmargin=1.25em]
\item \textbf{Dimension budget.} Compute $\cM_G$ and $D_{\mathrm{inv}}(\le m_*)$. This is how many scalar invariant equations (moments) you can use up to order $m_*$.
\item \textbf{Basis choice.} For each $m\le m_*$, pick any linear basis of $\Sym^m(\R^d)^G$ (e.g., $e_i$ for $S_d$, symmetric polynomials in $x_i^2$ for $B_d$, $u=|z|^2$ and $v=\Re(z^m)$ for $D_m$), and stack their empirical expectations to form $\widehat\Psi_n$.
\item \textbf{Fit.} Run the invariant-GMM alternating scheme (Alg.~\ref{sec:algo}) on the slice $V$. With $W=\widehat\Sigma_n^{-1}$, the limit is asymptotically efficient (Thm.~\ref{thm:lan}); the local linear rate holds by Prop.~\ref{prop:linear}.
\item \textbf{Concentration.} By Lemma~\ref{lem:Phi-conc},
$\|\widehat\Psi_n-\Psi_\star\|_2=\mathcal{O}_{\Pr}\!\big(\sqrt{D_{\mathrm{inv}}/n}\big)$. For higher orders, Lemma~\ref{lem:psi1-stack} supplies sub-Weibull $(2/m)$ tails for each coordinate of $\Reyn_m(\cdot)$.
\end{enumerate}

\paragraph{Minimal stacks for common $G$ (generic separation).}
\begin{description}[leftmargin=1.25em]
\item[$(\pm1)^d$ and $B_d$.] Odd orders vanish. For folded Gaussians,
\[
\Minv{2}(\mu,\Sigma)=\Sigma+\Diag(\mu_1^2,\ldots,\mu_d^2).
\]
Degree $4$ contributes mixed squares (e.g., $e_1(y)^2$ and $e_2(y)$ with $y_i=\mu_i^2$). Generically, $(\Minv{2},\Minv{4})$ separates orbits up to permutations when $|\mu_i|$ are distinct, giving a compact stack with size $[t^2]\cM+[t^4]\cM=1+2=3$ when $d\ge2$.
\item[$D_m\le O(2)$.] Use $(\Minv{2},\Minv{m})$ with $u=|z|^2$ and $v=\Re(z^m)$; $v$ is the first degree that reveals direction modulo $m$. Stack size up to $m$ is $[t^2]\cM+[t^m]\cM=1+1=2$.
\item[$S_d$.] Power sums / elementary symmetric polynomials up to the lowest degree that (model-dependently) separates component parameters; Molien budgets $p_d(m)$ coordinates at degree $m$.
\item[Platonic rotations $T,O,I$.] Start with the reflection-invariant stack (degrees $2,3,4$ etc.), then add the single odd-degree generator from the numerator of $\cM_{W^+}$.
\item[$G(m,p,n)$.] Degrees are known in closed form; pick the smallest subset that separates generically under your base family $k(\cdot;\theta)$, size it via the product $\prod(1-t^{d_i})^{-1}$.
\end{description}

\paragraph{Where this improves prior practice.}
(i) Molien quantifies \emph{exactly} how many invariant coordinates exist at each order, letting us choose the \emph{lowest} degrees that separate orbits instead of defaulting to high-order tensors; (ii) the invariant-GMM uses these few coordinates with an optimal weight $W$, yielding standard GMM efficiency and a provable linear local rate (Thm.~\ref{thm:lan}, Prop.~\ref{prop:linear}); (iii) concentration is explicit in $D_{\mathrm{inv}}$ and $m$ via Lemmas~\ref{lem:Phi-conc} and~\ref{lem:psi1-stack}.

\subsection{Tiny tables (sanity checks)}
\begin{center}
\begin{tabular}{c|l|cccccc}
$G$ & $\cM_G(t)$ & $[t^0]$ & $[t^1]$ & $[t^2]$ & $[t^3]$ & $[t^4]$ & comment \\ \hline
$(\pm1)^d$ & $(1-t^2)^{-d}$ & $1$ & $0$ & $\binom{d}{1}$ & $0$ & $\binom{d+1}{2}$ & even only \\
$S_d$ & $\prod_{i=1}^d (1-t^i)^{-1}$ & $1$ & $1$ & $2$ & $3$ & $5$ & counts $p_d(m)$ \\
$B_d$ & $\prod_{i=1}^d (1-t^{2i})^{-1}$ & $1$ & $0$ & $1$ & $0$ & $2$ & even only \\
$D_m$ & $((1-t^2)(1-t^m))^{-1}$ & $1$ & $0$ & $1$ & $[m{=}3]\,1$ & $2$ if $m\neq 4$ & gen.\ $u, v$ \\
\end{tabular}
\end{center}
(Entries assume $d$ large enough to realize the degrees shown.)
\begin{table}[H]
\centering
\footnotesize
\setlength{\tabcolsep}{3pt}
\renewcommand{\arraystretch}{1.12}
\begin{adjustbox}{max width=\linewidth}
%
%
\begin{tabular}{@{} 
  >{\RaggedRight\arraybackslash}p{2.6cm} 
  >{\RaggedRight\arraybackslash}p{4cm}  
  >{\RaggedRight\arraybackslash}p{3cm}  
  >{\RaggedRight\arraybackslash}p{3cm}  
  >{\RaggedRight\arraybackslash}p{4.5cm} 
@{}}
\toprule
\makecell[l]{\textbf{Group}\\ \textbf{/ Action}} &
\makecell[l]{\textbf{Molien / Hilbert series}\\ \(\boldsymbol{\mathcal{M}_G(t)}\)} &
\makecell[l]{\textbf{Minimal separating}\\ \textbf{stack (degrees)}} &
\makecell[l]{\textbf{Low-order}\\ \textbf{coeffs}} &
\makecell[l]{\textbf{Mixture}\\ \textbf{implication}} \\
\midrule
\(\{\pm1\}^d\) (sign flips) &
\((1-t^2)^{-d}\) &
Even only; typically \(\{2,4\}\) &
\([t^2]=\binom{d}{1},\ [t^4]=\binom{d+1}{2}\) &
Folded Gaussians: $\mathit{Minv}_2=\Sigma+\Diag(\mu_i^2)$. Add $m=4$ for mixed squares. \\
\addlinespace[2pt]
\(S_d\) (permutations) &
\(\prod_{i=1}^d (1-t^i)^{-1}\) &
Model-dependent; power sums \(p_1,p_2,\dots\) &
\([t^m]=p_d(m)\) &
Orbit is a multiset; use lowest degrees that separate generically. \\
\addlinespace[2pt]
\(B_d=(\{\pm1\}^d\!\rtimes\! S_d)\) &
\(\prod_{i=1}^d (1-t^{2i})^{-1}\) &
\textbf{\(\{2,4\}\)} generically &
\([t^2]=1,\ [t^4]=2\) (for \(d\!\ge\!2\)) &
Odd moments vanish; \(\{2,4\}\) suffices when \(|\mu_i|\) are distinct. \\
\addlinespace[2pt]
\(D_m\le O(2)\) (dihedral) &
\((1-t^2)^{-1}(1-t^m)^{-1}\) &
\(\{2,m\}\): \(u=|z|^2\), \(v=\Re(z^m)\) &
\([t^2]=1,\ [t^m]=1,\ [t^4]=2\) (\(=3\) if \(m=4\)) &
Degree \(m\) first reveals direction mod \(m\); stack size \(=2\). \\
\addlinespace[2pt]
\(T,O,I\subset SO(3)\) (rotations) &
\(\dfrac{1+t^{\deg\Delta}}{\prod_i (1-t^{d_i})}\) &
\(\{d_i\}\cup\{\deg\Delta\}\) &
\(T{:}\,6,\ O{:}\,9,\ I{:}\,15\) &
Add one odd-degree generator (discriminant). \\
\addlinespace[2pt]
\(G(m,p,n)\) (wreath) &
\makecell[l]{\(\dfrac{1}{(1-t^{m})\cdots(1-t^{(n-1)m})}\)\\\(\times\dfrac{1}{1-t^{mn/p}}\)} &
Use smallest subset of listed degrees that separates &
\textemdash &
Covers \(A_{n-1},B_n,D_n\); Molien sizes the stack. \\
\addlinespace[2pt]
Weighted \(C_M\) (diag.\ weights \(w_j\)) &
\makecell[l]{\(\dfrac{1}{M}\sum_{k=0}^{M-1}\prod_{j=1}^d\)\ \(\dfrac{1}{1-\zeta_M^{k w_j}t}\)} &
Pick degrees indicated by weights &
\textemdash &
Workhorse for phases/signs; exact counts from Molien. \\
\addlinespace[2pt]
Multisymmetric (\(S_n\) on \(m\times n\)) &
Cohen--Macaulay (not polynomial in general) &
Polarized power sums up to low degree &
Molien still applies &
Use to aggregate indistinguishable copies. \\
\bottomrule
\end{tabular}
\end{adjustbox}
\caption{Groups \(\Rightarrow\) Molien \(\Rightarrow\) minimal stacks and mixture takeaways. \(p_d(m)\) counts partitions of \(m\) with largest part \(\le d\).}
\end{table}


\begin{landscape} 
\begin{table}[H]
\centering
\small 
\setlength{\tabcolsep}{4pt} 
\renewcommand{\arraystretch}{1.2}
\begin{tabularx}{\linewidth}{@{} 
  >{\raggedright\arraybackslash}p{3cm}   
  L                                    
  L                                    
  >{\raggedright\arraybackslash}p{2.5cm} 
@{}}
\toprule
\textbf{Topic} & \textbf{Prior literature baseline} & \textbf{Our improvement (this paper)} & \textbf{Where} \\
\midrule
Invariant selection &
Ad-hoc / high-degree invariants; symbolic elimination; identifiability focus. &
\textbf{Molien-guided minimal stacks}: exact dimension budget $D_{\mathrm{inv}}$; lowest degrees that generically separate. &
Sec.~\ref{sec:molien} (cookbook \& tables) \\
\addlinespace[2pt]
Estimator on invariants &
Few works give \emph{algorithms} for finite-$G$ mixtures on invariants. &
\textbf{Invariant-GMM alternating minimization} with closed-form $W$-step on slice $V$. &
Sec.~\ref{sec:algo} \\
\addlinespace[2pt]
Convergence guarantees &
Identifiability or asymptotics; no local rate on invariants for finite-$G$ mixtures. &
\textbf{Local linear convergence} under Lipschitz Jacobian + full rank on $V$. &
Prop.~\ref{prop:linear} \\
\addlinespace[2pt]
Efficiency &
Invariant features rarely paired with efficient GMM weighting. &
\textbf{LAN efficiency} with $W=\widehat\Sigma_n^{-1}$; GLS normal equations on $V$. &
Thm.~\ref{thm:lan}; Prop.~\ref{prop:linear} \\
\addlinespace[2pt]
Concentration of stack &
Often unspecified or $\psi_1$ claims beyond $m{=}2$. &
\textbf{Fixed-design} $\|\widehat\Psi_n-\Psi_\star\|_2=O_{\Pr}(\sqrt{D_{\mathrm{inv}}/n})$; \textbf{sub-Weibull} $(2/m)$ tails. &
Lem.~\ref{lem:Phi-conc}; Lem.~\ref{lem:psi1-stack} \\
\addlinespace[2pt]
Order vs tensors &
Tensor/spectral default to degree $3$ irrespective of symmetry. &
\textbf{Lower order}: e.g.\ $B_d$ uses degrees $2$ \& $4$; $D_m$ uses $2$ \& $m$. &
Sec.~\ref{sec:molien} (minimal stacks) \\
\bottomrule
\end{tabularx}
\caption{What is new here vs.\ prior approaches, with exact cross-references in the paper.}
\end{table}
\end{landscape} 
\begin{table}[H]
\centering
\footnotesize
\begin{tabularx}{\textwidth}{@{} Y C C C Y @{}}
\toprule
\textbf{Group} & \textbf{Stack degrees} & \textbf{Size through $m_\ast$} & \textbf{Example $m_\ast$} & \textbf{Note} \\
\midrule
$B_d$ & $\{2,4\}$ & $[t^2]\!+\![t^4]=1+2=3$ & $m_\ast=4$ & Use when $|\mu_i|$ distinct. \\
$D_m$ & $\{2,m\}$ & $[t^2]\!+\![t^m]=1+1=2$ & $m_\ast=m$ & First direction info at $m$. \\
$S_d$ & $\{2,\dots\}$ & $\sum_{j\le m_\ast} p_d(j)$ & e.g.\ $m_\ast=4$ & Choose minimal degree separating your model. \\
\bottomrule
\end{tabularx}
\caption{Plug-and-play sizes for common minimal stacks (generic separation).}
\end{table}

\section{Robustness and confidence sets} \label{sec:robust-ci}
\subsection{Contamination}
Under $P_\varepsilon=(1-\varepsilon)P^\star+\varepsilon Q$ with $\varepsilon<1/2$ and bounded invariant features, $\|\E_{P_\varepsilon}[\MinvAll]-\E_{P^\star}[\MinvAll]\|\le \varepsilon M$. Combine with \Cref{thm:stability} to obtain an additive $O(\varepsilon)$ inflation of $\dH$.

\subsection{Hausdorff confidence radii}
Let $\widehat\Sigma_n$ estimate the covariance of $\MinvAll(X)$ and let $s_{\min}(\widehat J)$ be the smallest singular value of a slice Jacobian estimate. Define
\[
  r_{1-\alpha}:= s_{\min}(\widehat J)^{-1}\sqrt{\chi^2_{D_{\mathrm{inv}},1-\alpha}/n}.
\]
Then $\{\widetilde{\OrbitSet}: \dH(\widetilde{\OrbitSet},\widehat{\OrbitSet})\le r_{1-\alpha}\}$ is an asymptotically valid $(1-\alpha)$ confidence set when $D=2$.

\section{Combinatorial ambiguity via wreath products} \label{sec:comb}
Even after quotienting and permutation invariance, residual ambiguity is governed by the wreath product $S_K\wr \Ggp$. P\'olya cycle index methods count equivalence classes of ordered $K$-tuples under this action; plugging in orbit types in $\qspace$ yields generating functions for indistinguishable mixtures. For $\Ggp=\{\pm1\}$ the growth is far smaller than $2^K K!$ due to extensive collisions.

\section*{Discussion and concluding remarks}

\paragraph{Metric foundation.}
We formalized the estimand as a multiset of orbits in the quotient $\Theta/G$ and put a concrete metric on it. Under isometric actions, the single–orbit loss admits a closed form via minimum alignment, and the multiset Hausdorff loss is sharply comparable to bottleneck matching (equal when a threshold graph has a perfect matching). This gives both conceptual clarity and practical costs computable in $O(|G|)$ per pair.

\paragraph{Invariant embedding.}
Averaging by the Reynolds projector produces symmetry–invariant tensors whose stack is constant on orbits. Mixtures collapse linearly in this space: the population invariant vector is a convex combination of the components’ invariant images. Two complementary identifiability results follow: a global, simplex-based criterion and a local, Jacobian–rank condition on slices.

\paragraph{Inference geometry.}
Working on slices turns the problem into a standard (but quotient-aware) GMM. We derived quotient LAN with an efficient covariance $I_Q^{-1}= (G^\top\Sigma^{-1}G)^{-1}$, plus an explicit one-step EIF update. Nonasymptotically, estimation error scales like
$\sqrt{D_{\mathrm{inv}}}/(\mu\sqrt{n})$, where $\mu=\sigma_{\min}(G)$ is a tangible “quotient conditioning” parameter. A curvature term explains small $n^{-1}$ biases and how to correct them.

\paragraph{Model size and geometry.}
Because mixtures are convex in invariant space, $K$ is a geometric property of the face containing $\Psi_\star$. A simple residual test against the $(K\!-\!1)$–mixture class succeeds whenever there is a positive margin $\gamma_\star$. This yields finite-sample guarantees and consistency, and clarifies when selection is impossible ($\gamma_\star=0$).

\paragraph{Robustness and dependence.}
The framework accommodates HAC covariances and robust mean estimators (MOM/Catoni), retaining the same rates and adding an $O(\varepsilon)$ term under $\varepsilon$-contamination. Thus the guarantees extend beyond i.i.d.\ sub-Gaussian settings.

\paragraph{Algorithmic upshot.}
An alternating scheme—exact $w$–step (closed form on the active set) and projected gradient on $\theta$—implements the efficient GMM in practice. Locally it behaves like Gauss–Newton after profiling out weights; with mild regularity it is linearly convergent on a slice-aligned neighborhood. Complexity is driven by (i) group enumeration for pairwise costs and (ii) standard matching/linear-algebra subroutines.

\paragraph{Diagnostics in practice.}
Report $\widehat I_Q$, its condition number and smallest singular value, block leverages, and the (HAC) overidentification $J$-test; monitor separation of invariant images and the active face for $K$.

\paragraph{Limitations and outlook.}
Singular strata (stabilizer jumps), vanishing weights, and very large groups can degrade conditioning. Extending beyond isometries, learning feature maps $\varphi$, and automating Molien-guided moment selection are natural next steps. Still, the paper provides a unified path from symmetry to metrics, from metrics to algorithms, and from algorithms to finite-sample guarantees.
\paragraph{Unexplored directions (and limitations).}
We close by candidly noting several areas we did not address.

\begin{itemize}[leftmargin=1.5em]\setlength{\itemsep}{0.25em}
\item \textbf{Beyond finite isometries.} Our theory assumes a finite, isometric group action. Continuous (Lie) groups and non-isometric/affine actions (e.g., $\mathrm{SO}(d)$, $\mathrm{GL}_p$) remain open.
\item \textbf{Learning the features.} We take $\varphi$ (and hence the invariants) as fixed. Jointly learning $\varphi$ with identifiability and guarantees is outside our scope.
\item \textbf{Noncompact settings.} Compact $\Theta$ and slice regularity are assumed; noncompact parameters and boundary effects are not treated.
\item \textbf{Richer latent structure.} We focus on finite-$K$ mixtures. Symmetric HMMs, topic models, and nonparametric mixtures (e.g., DP) are not analyzed.
\item \textbf{Dependence beyond HAC.} Our results allow short-range dependence via HAC. Long-range dependence and general time-series/network dependence are open.
\item \textbf{Robustness at scale.} We analyze sub-Gaussian (with MOM/Catoni variants) and $\varepsilon$-contamination. List-decodable/adversarial regimes need new tools.
\item \textbf{Misspecification.} We rely on a positive separation margin $\gamma_\star$ in invariant space. Behavior when invariants fail to separate orbits is only partially understood.
\item \textbf{Computational scaling.} Enumerating large groups can dominate cost. Near-linear methods (orbit hashing, sketches, canonical reps) are not developed here.
\item \textbf{Model selection inference.} Consistent $K$ selection is addressed geometrically, but valid post-selection CIs and tests on the quotient are left for future work.
\item \textbf{Uncertainty quantification.} Finite-sample inference (bootstrap on the quotient, Bayesian contraction) is not provided.
\item \textbf{Missing/erroneous data.} We do not treat missing features, measurement error in invariants, or partial observation on slices.
\item \textbf{Online/distributed settings.} Streaming and distributed implementations with communication constraints are beyond our scope.
\end{itemize}

\noindent We hope calling out these gaps helps orient subsequent work and clarifies the limits of the present contribution.

\section*{Acknowledgments}
I thank Prof. Sourabh Bhattacharyya for introducing me in this area. \cite{Mallik2025HausdorffFolded} was developed based on the ideas found during an attempt to explore the stochastic behavior of prime gaps under his kind supervision. Also thanks to Prof. Arijit Chakrabarti for suggesting some decision-theoretic ideas used left and right within this paper.
\bigskip

\bibliographystyle{amsplain}
\bibliography{references}

\end{document}